\documentclass[11pt]{article}
\usepackage{fullpage,comment,algorithm,algorithmic}
\usepackage{color,colortbl}
\usepackage{float}
\usepackage{amsthm}
\usepackage{amsmath}
\usepackage{amssymb}
\usepackage{graphicx}
\usepackage{xfrac}
\usepackage[colorinlistoftodos,prependcaption,textsize=tiny]{todonotes}
\usepackage{multirow}
\usepackage{wrapfig}
\usepackage{boldline}
\usepackage{thm-restate}

\usepackage[colorlinks,linkcolor=blue,filecolor=blue,citecolor=blue,urlcolor=blue,pdfstartview=FitH,pagebackref]{hyperref}
\usepackage[nameinlink]{cleveref}

\newtheorem{theorem}{Theorem}
\newtheorem{lemma}{Lemma}
\newtheorem*{lemma*}{Lemma}
\newtheorem{corollary}[lemma]{Corollary}
\newtheorem{remark}{Remark}
\newtheorem{claim}[lemma]{Claim}

\newtheorem{definition}{Definition}

\newcommand{\namedref}[2]{\hyperref[#2]{#1~\ref*{#2}}}
\newcommand{\sectionref}[1]{\namedref{Section}{#1}}

\newcommand{\theoremref}[1]{\namedref{Theorem}{#1}}

\newcommand{\lemmaref}[1]{\namedref{Lemma}{#1}}
\newcommand{\tableref}[1]{\namedref{Table}{#1}}

\newcommand{\ddim}{{\rm ddim}}

\newcommand{\E}{{\mathbb{E}}}
\newcommand{\N}{\mathbb{N}}
\newcommand{\R}{\mathbb{R}}

\newcommand{\poly}{{\rm poly}}

\newcommand{\argmax}{\rm argmax}

\newcommand{\eps}{\epsilon}

\newcommand{\rt}{rt}
\newcommand{\bp}{\rm BP}
\newcommand{\slt}{\rm SLT}
\newcommand{\hd}{hop-diameter }

\newcommand{\mL}{\mathcal{L}}
\newcommand{\F}{\mathcal{F} }

\newcommand{\D}{\Delta}

\def\inline#1:{\par\vskip 7pt\noindent{\bf #1:}\hskip 10pt}

\def\inline#1:{\par\vskip 7pt\noindent{\bf #1:}\hskip 10pt}

\def\blackslug{\hbox{\hskip 1pt \vrule width 4pt height 8pt
		depth 1.5pt \hskip 1pt}}

\def\QED{\quad\blackslug\lower 8.5pt\null\par}

\newcommand{\alert}[1]{\textbf{\color{red}
[[[#1]]]}\marginpar{\textbf{\color{red}**}}\typeout{ALERT:
\the\inputlineno: #1}}

\floatstyle{ruled}
\newfloat{algorithm}{tbp}{loa}
\providecommand{\algorithmname}{Algorithm}
\floatname{algorithm}{\protect\algorithmname}

\definecolor{forestgreen}{rgb}{0.13, 0.55, 0.13}

\makeatother

\usepackage{authblk}
\usepackage{pdfsync}

\begin{document}
\date{}
\author{Michael Elkin\thanks{Supported in part by ISF grant (724/15).}}
\author{Arnold Filtser\thanks{Supported in part by ISF grant (1817/17) and in part by BSF grant 2015813}}
\author{Ofer Neiman\thanks{Supported in part by ISF grant (1817/17) and in part by BSF grant 2015813}}
\affil{Ben-Gurion University of the Negev. Email: \texttt{\{elkinm,arnoldf,neimano\}@cs.bgu.ac.il}}
\title{Distributed Construction of Light Networks}
\maketitle
\begin{abstract}
A $t$-{\em spanner} $H$ of a weighted graph $G=(V,E,w)$ is a subgraph that approximates all pairwise distances up to a factor of $t$. The {\em lightness} of $H$ is defined as the ratio between the weight of $H$ to that of the minimum spanning tree. An $(\alpha,\beta)$-{\em Shallow Light Tree} (SLT) is a tree of lightness $\beta$, that approximates all distances from a designated root vertex up to a factor of $\alpha$. A long line of works resulted in efficient algorithms that produce (nearly) optimal light spanners and SLTs.

Some of the most notable algorithmic applications of light spanners and SLTs are in distributed settings. Surprisingly, so far there are no known efficient distributed algorithms for constructing these objects in general graphs. In this paper we devise efficient distributed algorithms in the CONGEST model for constructing light spanners and SLTs, with near optimal parameters. Specifically, for any $k\ge 1$ and $0<\epsilon<1$, we show a $(2k-1)\cdot(1+\epsilon)$-spanner with lightness $O(k\cdot n^{1/k})$ can be built in $\tilde{O}\left(n^{\frac12+\frac{1}{4k+2}}+D\right)$ rounds (where $n=|V|$ and $D$ is the hop-diameter of $G$). In addition, for any $\alpha>1$ we provide an $(\alpha,1+\frac{O(1)}{\alpha-1})$-SLT in $(\sqrt{n}+D)\cdot n^{o(1)}$ rounds. The running time of our algorithms cannot be substantially improved.

We also consider spanners for the family of doubling graphs, and devise a $(\sqrt{n}+D)\cdot n^{o(1)}$ rounds algorithm in the CONGEST model that computes a $(1+\epsilon)$-spanner with lightness $(\log n)/\epsilon^{O(1)}$. As a stepping stone, which is interesting in its own right, we first develop a distributed algorithm for constructing nets (for arbitrary weighted graphs), generalizing previous algorithms that worked only for unweighted graphs.
\end{abstract}

\section{Introduction}

Let $G=(V,E,w)$ be a graph with edge weights $w:E\to\R_+$. For $u,v\in V$, denote by $d_G(u,v)$ the shortest path distance in $G$ between $u,v$ with respect to these weights. A subgraph $H=(V,E')$ with $E'\subseteq E$ is called a {\em $t$-spanner} of $G$, if for all $u,v\in V$, $d_H(u,v)\le t\cdot d_G(u,v)$. The parameter $t$ is called the {\em stretch} of $H$. The most relevant and studied attributes of a $t$-spanner are its sparsity (i.e., the number of edges $|E'|$), and the total weight of the edges $w(H)=\sum_{e\in E'}w(e)$. 
Since any spanner with finite stretch must be connected, its weight is at least the weight of the Minimum Spanning Tree (MST) of $G$, and the {\em lightness} of $H$ is defined as $\frac{w(H)}{w(MST)}$.

Another useful notion of a subgraph that approximately preserves distances was introduced in \cite{ABP92,KRY95}. Given a weighted graph $G=(V,E,w)$ with a designated root vertex $\rt$, an $(\alpha,\beta)$-{\em Shallow-Light Tree} (SLT) $T$ of $G$ is a spanning tree which has lightness $\beta$, and approximates all distances from $\rt$ to the other vertices up to a factor of $\alpha$.

In this paper we focus on the distributed CONGEST model of computation, where each vertex of the graph $G$ hosts a processor, and these processors communicate with each other in discrete rounds via short messages on the graph edges (typically the message size is $O(\log n)$ bits). We devise efficient distributed algorithms that construct light spanners and shallow-light trees for general graphs, and also light spanners for doubling graphs. See \tableref{fig:table} for a succinct summary.

\begin{table*}
\centering
\small
\begin{tabular}{|c|c|c|c|c|}
	\hline
     Object  &  Distortion   &  Lightness & Size & Run time \\
	\hlineB{3}
Spanner & $(2k-1)\cdot(1+\epsilon)$ & $O(k\cdot n^{1/k})$ & $O(k\cdot n^{1+1/k})$ & $\tilde{O}\left(n^{\frac12+\frac{1}{4k+2}}+D\right)$\\
	\hline
SLT & $1+\frac{O(1)}{\alpha-1}$ & $\alpha$ & NA & $\tilde{O}\left(\sqrt{n}+D\right)\cdot\poly\left(\frac{1}{\alpha-1}\right)$\\
\hline
$(\gamma,\beta)$-net & NA & NA & NA & $(\sqrt{n}+D)\cdot 2^{\tilde{O}(\sqrt{\log n}\cdot\log\frac{\beta}{\gamma-\beta})}$\\
\hline
Spanner & $1+\epsilon$ &  $\eps^{-O(\ddim)}\cdot \log n$  & $n\cdot \eps^{-O(\ddim)}\cdot\log n$ &  $(\sqrt{n}+D)\cdot\eps^{-\tilde{O}(\sqrt{\log n}+\ddim)}$\\
\hline
\end{tabular}
\caption{A summary of our main results. Here $n$ is the number of vertices, $k\ge 1$ and $\alpha\ge 1 $ are parameters and $0<\epsilon<1$ is a constant. For the nets $\gamma>\beta>0$. All results are in the CONGEST model.}\label{fig:table}
\end{table*}

\subsection{Light Spanners for General Graphs}

Spanners are a fundamental combinatorial object. They have been extensively studied and have found numerous algorithmic applications \cite{A85,PS89,PU89,ADDJS93,Coh98,ACIM99,EP01,BS07,E06,EZ06,TZ06,Pet09,DGPV08,P10,MPVX15,AB16,EN17}. The basic greedy algorithm \cite{ADDJS93}, for a graph with $n$ vertices and any integer $k\ge 1$, provides a $(2k-1)$-spanner with $O(n^{1+1/k})$ edges, which is best possible (assuming Erdos' girth conjecture). Spanners of low weight have received much attention in recent years \cite{CDNS95,ES16,ENS15,G15,CW18,BFN16,ADFSW17,BLW17,FN18,BLW19}, and are particularly useful in a distributed setting;  efficient broadcast protocols, network synchronization and computing global functions \cite{ABP90,ABP92}, network design \cite{MP98,SCRS01} and routing \cite{WCT02} are a few examples. The state-of-the-art is a $(2k-1)\cdot(1+\epsilon)$-spanner of \cite{CW18} with lightness $O(n^{1/k})$, for any constant $0<\epsilon<1$. In \cite{FS16} it was shown that the greedy algorithm is existentially optimal, hence it also achieves such lightness.

The greedy algorithm provides a satisfactory answer to the existence of sparse and light spanners, but not to efficiently producing such a spanner (because the greedy algorithm has inherently large running time). Indeed, the problem of devising fast algorithms to construct spanners is very important in some algorithmic applications. For light spanners, \cite{ES16} showed a near-linear time algorithm that constructs a $(2k-1)\cdot(1+\epsilon)$-spanner with $O(k\cdot n^{1+1/k})$ edges and lightness $O(k\cdot n^{1/k})$. The sparsity and lightness were improved (still in near-linear time) in \cite{ADFSW17} to $O(\log k\cdot n^{1+1/k})$ and $O(\log k\cdot n^{1/k})$ respectively. In the distributed setting, \cite{BS07} devised a randomized algorithm for a $(2k-1)$-spanner with $O(k\cdot n^{1+1/k})$ edges in $O(k)$ rounds in the CONGEST model. This was recently improved for unweighted graphs by \cite{MPVX15,EN17spanner} to $O(n^{1+1/k})$ edges.
 However, the weight of  these spanners is not bounded. Surprisingly, none of the previous works in the CONGEST model has a bound on the {\em lightness} of spanners for general graphs.

\paragraph{Our results.}
Unlike the sparsity of spanners, which can be preserved via a local algorithm, the lightness is a global measure. Indeed, we observe that the lower bound of \cite{SHKKNPPW12} on the number of rounds required for any polynomial approximation of the MST weight, implies a lower bound for computing light spanners. In particular, for a graph with $n$ vertices and hop-diameter\footnote{The hop diameter of a weighted graph is the diameter of the underlying unweighted graph.} $D$, any CONGEST algorithm requires at least $\tilde{\Omega}(\sqrt{n}+D)$ rounds for computing a light spanner (with any polynomial lightness).\footnote{The notations $\tilde{O}(\cdot)$ and $\tilde{\Omega}(\cdot)$ hide polylogarithmic factors.}

We provide the first algorithm with sub-linear number of rounds for constructing light spanners for general graphs in the CONGEST model. Specifically, for any integer parameter $k\ge 1$ and constant $0<\epsilon<1$, we devise a randomized algorithm that w.h.p. outputs a $(2k-1)\cdot(1+\epsilon)$-spanner with $O(k\cdot n^{1+1/k})$ edges and lightness $O(k\cdot n^{1/k})$, within $\tilde{O}\left(n^{\frac12+\frac{1}{4k+2}}+D\right)$ rounds in the CONGEST model, thus nearly matching the lower bounds.

\subsection{Shallow-Light Trees}

Shallow-Light trees are widely used for various distributed tasks, such as network design, broadcasting in ad-hoc networks and multicasting \cite{PV04,SDS04,YCC06}.
In \cite{KRY95}, an optimal tradeoff between the lightness of the SLT to the stretch of the root distances was obtained. Specifically, for any $\alpha>1$ they obtained an SLT with lightness $\alpha$ and stretch $1+\frac{2}{\alpha-1}$. In addition, \cite{KRY95} exhibited an efficient algorithm for constructing such a tree in near-linear time, and also in $O(\log n)$ rounds in the PRAM (CREW) model. However, their techniques are inapplicable to the CONGEST model, and it remained an open question whether an SLT can be built efficiently in this model. (Roughly speaking, \cite{KRY95} used pointer jumping techniques that require communication between non-adjacent vertices, hence this is unsuitable for the CONGEST model.)

\paragraph{Our result.} Here we answer this question, and devise a distributed deterministic algorithm, that for any $\alpha>1$, outputs an SLT with lightness $\alpha$ and stretch $1+\frac{O(1)}{\alpha-1}$, within $\tilde{O}\left(\sqrt{n}+D\right)\cdot\poly\left(\frac{1}{\alpha-1}\right)$ rounds. Once again, any distributed SLT algorithm must take at least $\tilde{\Omega}(\sqrt{n}+D)$ rounds \cite{E04,SHKKNPPW12}. Thus our result is nearly optimal.

\subsection{Light Spanners for Doubling Graphs}

A graph $G$ has \emph{doubling dimension} $\ddim$ if for every vertex $v\in V$ and radius $r>0$, the ball\footnote{A ball is defined as $B_G(v,r)=\{u\in V~:~d_G(u,v)\le r\}$.} $B_G(v,2r)$ can be covered by $2^\ddim$ balls of radius $r$. For instance, a $d$-dimensional $\ell_p$ space has $\ddim=\Theta(d)$, and every graph with $n$ vertices has $\ddim=O(\log n)$.
This is a standard and well-studied notion of "growth restriction" on a graph \cite{Ass83,GKL03,HPM06}, and it is believed that such graphs occur often in real-life networks and data \cite{TSL00,NZ02}. One notable motivation for light spanners in doubling graphs\footnote{A graph family is called {\em doubling} if its members have constant doubling dimension.} is their application for polynomial approximation schemes for the traveling salesperson and related problems (see, e.g., \cite{K05,G15}). While spanners with $1+\epsilon$ stretch and constant lightness have been known to exists in low dimensional Euclidean space for a while \cite{DHN93,ADDJS93}, only recently such $(1+\epsilon)$-spanners with constant lightness $(\ddim/\eps)^{O(\ddim)}$ have been discovered for doubling graphs \cite{G15}. The lightness was improved by \cite{BLW19} to the optimal $(1/\eps)^{O(\ddim)}$.

In the distributed LOCAL model\footnote{The LOCAL model is similar to CONGEST, but the size of messages is not bounded.}, \cite{DPP06} devised light spanners for a certain graph family, called {\em unit ball graphs in a doubling metric space}\footnote{A unit ball graph is a graph whose vertices lie in a metric space, and edges connect vertices of distance at most 1. In this scenario the metric is doubling.}. Specifically, they showed an $O(\log^*n)$ rounds algorithm for a $(1+\epsilon)$-spanner with lightness $(1/\eps)^{O(\ddim)}\cdot\log\Lambda$, where $\Lambda$ is the aspect ratio of $G$ (the ratio between the largest to smallest edge weights). We note that to obtain such a low number of rounds, they imposed restrictions on both the distributed model and the graph family.

Essentially all spanners for doubling graphs use {\em nets} in their construction. An {\em $(\alpha,\beta)$-net} of a graph is a set $N\subseteq V$ which is both $\alpha$-covering: for all $u\in V$ there is $v\in N$ with $d_G(u,v)\le\alpha$, and $\beta$-separated: for all $x,y\in N$, $d_G(x,y)>\beta$. The standard definition of a net is when $\alpha=\beta$, but we shall allow $\alpha>\beta$ as well. The usefulness of nets in doubling graphs stems from the fact that any net restricted to a ball of certain radius has a small cardinality. While a simple greedy algorithm yields a net, it is not suitable for distributed models due to it being inherently sequential.

A {\em ruling set} is a net in an unweighted graph. There have been several works that compute a ruling set in distributed settings.  In \cite{AGLP89}, a deterministic algorithm for a $(k\log n,k)$-ruling set running in $O(k\log n)$ rounds was developed, and a tradeoff extending this result was shown in \cite{SEW13}. A consequence of the work of \cite{AGLP89} provides a $(k,k)$-ruling set computed within $k\cdot 2^{\tilde{O}(\sqrt{\log n})}$ rounds. A randomized algorithm for a $(k,k)$-ruling set was given in \cite{L86} with $O(k\log n)$ rounds, and the running time was improved for graphs of small maximum degree in \cite{BEPS12,Gh16}.
However, all these results apply only for unweighted graphs. The problem of efficiently constructing a net in distributed models remained unanswered.

\paragraph{Our results.}
We design a randomized distributed algorithm, that for a given graph with $n$ vertices and hop-diameter $D$ and any $0<\beta< \alpha< 2\beta$, w.h.p. finds an $(\alpha,\beta)$-net within $(\sqrt{n}+D)\cdot 2^{\tilde{O}(\sqrt{\log n}\cdot\log\frac{\beta}{\alpha-\beta})}$ rounds in the CONGEST model. We show that the running time must be at least $\tilde{\Omega}(\sqrt{n}+D)$ for general graphs, via a reduction to the problem of approximating the MST weight.  So our running time is best possible (up to lower order terms). However, we do not know if a faster algorithm is achievable when the input graph has a constant doubling dimension.

Then, we utilize this algorithm for constructing nets, and devise a randomized algorithm that for a graph with doubling dimension $\ddim$ and any $0<\epsilon<1$, w.h.p. produces a $(1+\epsilon)$-spanner with lightness $\eps^{-O(\ddim)}\cdot \log n$ in $(\sqrt{n}+D)\cdot\eps^{-\tilde{O}(\sqrt{\log n}+\ddim)}$ rounds.

\subsection{Overview of Techniques}
In this section we provide an overview of the algorithms, techniques and ideas used in the paper. For the sake of brevity, some parts are over-simplified, or even completely neglected.

\paragraph{Eulerian Tour of the MST.}
Let $T$ be the MST. The first step in both our constructions of an SLT and a light spanner for general graphs is a distributed computation of a DFS traversal $\mL$ of $T$.
As an outcome of this computation, each vertex knows all its visiting times in $\mL$.
Our algorithm is a simplification of a similar algorithm from \cite{EN18}.

In order to compute $\mL$, we note that the distributed MST algorithm of \cite{KP98} induces a partition of $T$ into $O(\sqrt{n})$ fragments, each with hop-diameter $O(\sqrt{n})$.
We then create a virtual tree $T'$ whose vertices are the fragments. As $T'$ has only $O(\sqrt{n})$ vertices, it is possible to broadcast $T'$ to the entire graph.
We first compute a DFS tour locally in each fragment, and broadcast the $O(\sqrt{n})$ lengths of these tours. We use this information and the known structure of $T'$ to globally compute the DFS visit times for the "roots'' of the fragments. Finally, locally in each fragment, we extend this to a traversal of the entire tree $T$.

\paragraph{Shallow Light Tree (SLT).}
An SLT is a combination of the MST $T$ and a shortest path tree (SPT) rooted in $\rt$. As currently the fastest known  exact SPT algorithms \cite{GL18,E17} require more than $\tilde{O}(\sqrt{n} + D)$ rounds, we use instead an approximate SPT , $T'$.
Our basic strategy (following \cite{ABP92, KRY95}) is to choose a subset of vertices called break points (\bp).
Then we construct a subgraph $H$ by taking $T$, and adding to $H$ the unique path in $T'$ from $\rt$ to every break point $v\in\bp$. The SLT is computed as yet another approximate SPT rooted in $\rt$, but now using $H$ edges only.

Let $\mL=\{x_0,x_1,\dots,x_{2n-2}\}$ be an  Eulerian traversal of the MST $\mL$ (each vertex may appear several times).
Ideally, we would like to choose $\bp=\{x_0,x_{i_1},x_{i_2},\dots\}$ such that (1) every pair of consecutive points $x_{i_j},x_{i_{j+1}}\in\bp$ is far, specifically $d_{\mL}(x_{i_j},x_{i_{j+1}})>\eps\cdot d_G(\rt,x_{i_{j+1}})$, and (2) every node $x_q\in\mL$ has a nearby break point $x_{i_j}\in\bp$, specifically $d_{\mL}(x_{i_j},x_{q})\le \eps\cdot d_G(\rt,x_{q})$.
The first condition is used to bound the lightness, while the second condition is used to bound the stretch.

The choice of \bp~ described above can be easily performed in a greedy manner by sequentially traversing the nodes in $\mL$.  Unfortunately, we cannot implement this sequential algorithm efficiently in a distributed manner.
Instead, we break $\mL$ into $O(\sqrt{n})$ intervals, each containing at most $\sqrt{n}$ nodes. We add the first node in each interval to a temporary break point set $\bp'$. Using these temporary break points as an anchor, we perform the  sequential algorithm simultaneously in all intervals, and add (permanent) break points. Finally, we broadcast $\bp'$ to $\rt$, which performs a local computation in order to sparsify this set. Specifically, it chooses a subset of $\bp'$ to serve as  permanent break points  using the sequential algorithm, and broadcasts the chosen break points to the entire network. Intuitively, we are building a separate SLT for the set $\bp'$, which filters out some of its points.
Our analysis shows that this two-step choice of break points loses only a constant factor in the lightness.
However, this constant factor loss implies that obtaining the full tradeoff (i.e.,  lightness close to 1) cannot be done trivially as before. To remedy this, we apply a reduction from \cite{BFN16} adapted to the CONGEST model.

\paragraph{Light Spanner for General Graphs.}
Our basic approach is similar to the algorithms of \cite{CDNS95,ES16,ENS15}. We divide the graph edges into $O(\log n)$ buckets, according to their weight. Denote by $L$ the weight of the MST multiplied by 2 (for technical reasons).
In the lowest level, we have all the edges of weight at most $L/n$. For this bucket we simply use the distributed spanner of \cite{BS07} for weighted graphs. Even though the algorithm of \cite{BS07} provides an upper  bound only the sparsity of the spanner, we can use their construction as the weight of the edges in this bucket is sufficiently small.

Consider the $i$-th bucket $E_i$, where all edges have weight in $\left(\frac{L}{(1+\eps)^{i+1}},\frac{L}{(1+\eps)^{i}}\right]$.
We use the MST traversal $\mL$ to divide the graph into $O(\frac{(1+\eps)^i}{\eps})$ clusters of diameter $\frac{\eps\cdot L}{(1+\eps)^{i}}$.
Next, define an unweighted cluster graph $\mathcal{G}_i$ whose vertices are the clusters, and inter-cluster edges are taken only from $E_i$. Intuitively, the diameter of clusters is an $\epsilon$-fraction of the edge weights, so using the MST edges to travel inside clusters will increase the distance by at most a $(1+\epsilon)$ factor.

We then simulate the spanner algorithm of \cite{EN17spanner} for unweighted graphs on ${\cal G}_i$, and obtain a spanner ${\cal H}_i$. For every edge $e\in \mathcal{H}_i$, we add a corresponding edge $e'\in E_i$ to the final spanner. The main technical part is this simulation, which is difficult since the communication graph $G$ is not the graph ${\cal G}_i$ for which we want a spanner. We distinguish between small and large clusters; for the former we use $\mL$ to pipeline information inside the clusters, while for the latter we convergecast all the relevant information to a single vertex, and then broadcast the decisions made by this vertex to the entire graph. For this approach to be efficient we need to refine the partition to clusters, so that small clusters will have bounded hop-diameter, and also ensure there are few large clusters, to bound the convergecast and broadcast time.

\paragraph{Net Construction.}

Our algorithm for an $(\alpha,\beta)$-net imitates, on a high level, previous ruling sets algorithms (like \cite{L86,MRSZ11}).\footnote{In fact, these papers showed algorithms for Maximal Independent Sets (MIS), but a $(k,k)$-ruling set is an MIS for the graph $G^k$.} Ideally, the net construction works as follows. Initially, all vertices are active. In each round, sample a permutation $\pi$. Each (active) vertex $v$ which is the first in the permutation with respect to (w.r.t.) its $\beta$-neighborhood joins the net. Every vertex for which some vertex from its $\alpha$-neighborhood joined the net, becomes inactive. Repeat until all vertices become inactive.

In order to check whether a vertex is the first in the permutation w.r.t. its $\beta$-neighborhood, we use Least Element (LE) lists \cite{Coh97}. Given a permutation $\pi$, a vertex $u$ belongs to the LE list of a vertex $v$, if  $u$ is the first in the permutation among all the vertices at distance at most $d_G(v,u)$ from $v$. In particular, given the LE list of $v$, we can check whether it should join the net or not.
Efficient distributed computation of an LE list is presented in \cite{FL16}. However, rather than computing the list w.r.t. the graph $G$, \cite{FL16} compute LE list w.r.t. an auxiliary graph $H$ that approximates $G$ distances up to a $1+\eps$ factor. Fortunately, we can cope with the approximation by taking $\alpha>(1+\eps)\beta$.
Once we compute the lists and choose which vertices will be added to the net, we compute an (approximate) shortest path tree rooted in the net points. All vertices at distance at most $\alpha$ from net points become inactive. This concludes a single round. After $O(\log n)$ rounds all vertices become inactive w.h.p.. The running time is dominated by the LE lists computations.

We also provide a lower bound on the number of rounds required to construct a net, by a reduction to the problem of approximating the weight of an MST. 

\paragraph{Light Spanner for Doubling Metrics.}

The basic idea for constructing spanners for doubling metrics is quite simple and well known. For every distance scale $\D$, construct an $(\alpha,\beta)$-net $N_\D$ where $\alpha,\beta\approx \eps\D$, and connect by a shortest path every pair of net points at distance at most $\D$. The stretch bound follows standard arguments, based on the covering property of nets. To prove lightness, we will use a packing argument, stating that every net point has at most $\eps^{-O(\ddim)}$ other net points at distance $\D$, and every net point must contribute to the MST weight at least $\eps\cdot\D$.

The main issue is implementing this algorithm efficiently in the CONGEST model. An efficient distributed construction of nets was already described above, and the remaining obstacle is to connect nearby net points. The problem is that the shortest path between nearby net points may contain many vertices, and we cannot afford to add these sequentially. We resolve this issue by conducting a $\D$-bounded multi-source approximate shortest paths (from each net point) based on {\em hopsets}. Roughly speaking, a hopset is a set of (virtual) edges added to the graph, so that every pair has an approximate shortest path containing few edges. We use the path-reporting hopsets of \cite{EN16}, so that the actual paths are added to the spanner. The running time is indeed bounded: we use the packing property of nets to show that every vertex participates in a bounded number of such approximate shortest path computations.

\subsection{Organization}
In \Cref{sec:traversal} we devise an Eulerian traversal of the MST, which will be used in the following sections. In \Cref{sec:SLT} we present our distributed construction of an SLT. In \sectionref{sec:spanner} we show our light spanners for general graphs. The construction of nets for general graphs is shown in \sectionref{sec:nets}, and their application to light spanners for doubling graphs is in \sectionref{sec:doubling}. Finally, the lower bounds are in \sectionref{sec:lower}.

\section{Preliminaries}

Let $G=(V,E,w)$ be a weighted graph with $n$ vertices, and let $d_G$ be the induced shortest path metric with respect to the weights. We assume that the minimal edge weight is 1, and that the maximal weight is $\poly(n)$. For $v\in V$ denote by $N(v)=\{u\in V~:~ \{u,v\}\in E\}$ its set of neighbors, and by $N^+(v)=N(v)\cup\{v\}$.
For a set $C\subseteq V$, the induced graph on $C$ is $G[C]$. The {\em weak diameter} of $C$ is $\max_{u,v\in C}\{d_G(u,v)\}$ and its strong diameter is $\max_{u,v\in C}\{d_{G[C]}(u,v)\}$. The hop-diameter of $G$ is defined as its diameter while ignoring the weights.

In the CONGEST model of distributed computation, the graph $G$ represents a network, and every vertex initially knows only the edges incident on it. Communication between vertices occurs in synchronous {\em rounds}. On every round, each vertex may send a small message to each of its neighbors. Every message has size at most $O(\log n)$ bits. The time complexity is measured by the number of rounds it takes to complete a task (we assume local computation does not cost anything). Often the time depends on $n$, the number of vertices, and $D$, the {\em hop-diameter} of the graph.
The following lemma formalizes the broadcast ability of a distributed network (see, e.g., \cite{P00}).
\begin{lemma}\label{lem:pipe}
Suppose every $v\in V$ holds $m_v$ messages, each of $O(1)$ words\footnote{We assume a word size is $\log n$ bits.}, for a total of $M=\sum_{v\in V}m_v$. Then all vertices can receive all the messages within $O(M+D)$ rounds.
\end{lemma}
A Breadth First Search (BFS) tree $\tau$ of $G$ of \hd $D$ (ignoring the weights) can be computed in $O(D)$ rounds.  Since all  our algorithms have a larger running time, we always assume that we have such a tree at our disposal.

\section{Eulerian Tour of the MST}\label{sec:traversal}
\begin{wrapfigure}{r}{0.27\textwidth}
	\begin{center}
		\vspace{-35pt}
		\includegraphics[width=0.25\textwidth]{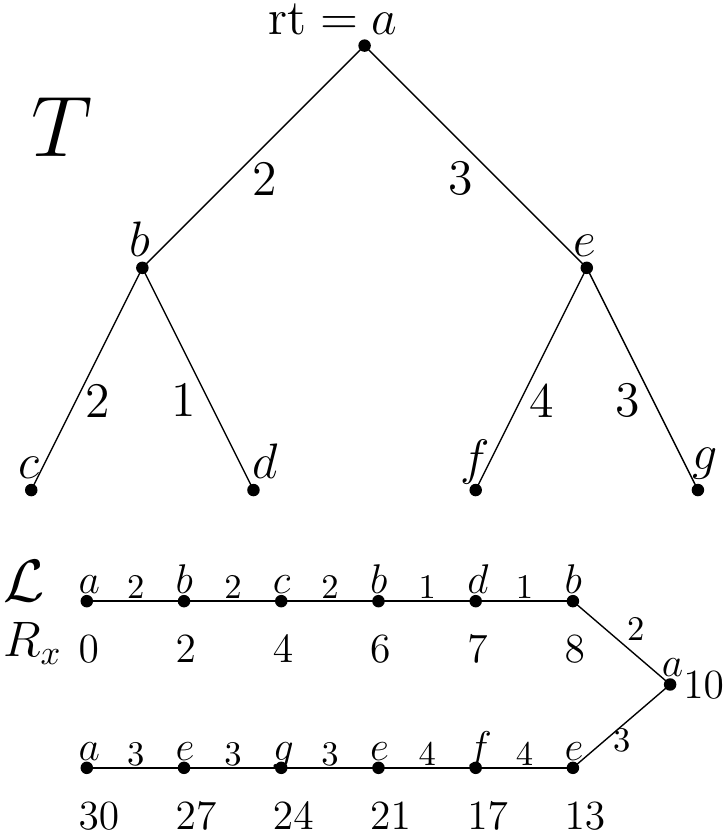}
		\vspace{-7pt}
	\end{center}
	\vspace{-30pt}
\end{wrapfigure}

Let $G=(V,E,w)$ be a weighted graph on $n$ vertices with \hd $D$. Let $T$ be the minimum spanning tree of $G$ with a root vertex $\rt\in V$. We compute an Eulerian path $\mL=\{\rt=x_0,x_1,\dots,x_{2n-2}\}$ drawn by taking a preorder traversal of $T$. The order between the children of a vertex is determined using their id. We remark that in \cite{EN18} it was described how to compute a DFS search of a tree in $\tilde{O}(\sqrt{n}+D)$ rounds. However, that paper also had the property that each vertex uses at most $O(\log n)$ words of memory. We give the full details here for completeness, and also since the presentation is somewhat simpler without the bound on the memory usage.

For a vertex $x\in\mL$, let $R_{x}=d_{\mL}(\rt,x)$ be the time visiting $x$ in $\mL$. The total length of the traversal $\mL$ (that is $R_{x_{2n-2}}$) equals $2\cdot w(T)$. The number of appearances of each vertex $v\in V$ in $\mL$ equals to its degree in $T$ (other that the root $\rt$ who has $\deg(\rt)+1$ appearances).
We will treat each such appearance as a separate vertex. That is $\mL$ is a path graph.
See figure on the right for an illustration.
For a vertex $v\in V$, let $\mL(v)\subseteq \mL$ be the set of appearances of $v$ in $\mL$.
In the remainder of this section we will prove the following lemma, that computes the traversal $\mL$ in $\tilde{O}(\sqrt{n}+D)$ rounds, meaning that each vertex $v$ will know $\mL(v)$ and the visiting time of every vertex $x\in \mL(v)$.

\begin{restatable}[MST traversal]{lemma}{MSTtraversal}\label{lem:MSTtraversal}	
	Let $G=(V,E,w)$ be a weighted graph with $n$ vertices, \hd $D$ and root $\rt\in V$, then there is a deterministic algorithm in the CONGEST model that computes $\mL$ in $\tilde{O}(\sqrt{n}+D)$ rounds.
\end{restatable}

\subsection{Computing the MST Fragments Tree}\label{sec:fragments}
In \cite{Elk17}, following \cite{KP98}, a deterministic MST construction in the CONGEST model with $\tilde{O}(\sqrt{n}+D)$ rounds was shown. 
The algorithm has two phases, according to the \hd of the fragments. At the end of the first phase, there is a set of $O(\sqrt{n})$ fragments $\F=\{F_1,F_2,\dots\}$, each with \hd $O(\sqrt{n})$. These fragments are called \emph{base fragments}.
The edges added in the first phase are called \emph{internal} edges (as each such edge is internal to some base fragment). In the second phase of the algorithm, the remaining $O(\sqrt{n})$ edges are added and connect the fragments to a tree. We call these edges \emph{external} edges, as they cross between base fragments.

Let $T'$ be a virtual tree with the base fragments $\F$ as vertices, and with the external edges as its edge set (i.e. there is an edge between $F_i$ and $F_j$ if there is an external edge between a vertex in $F_i$ to a vertex in $F_j$).
Since there are $O(\sqrt{n})$ vertices in $T'$, in $O(\sqrt{n}+D)$ rounds we can broadcast $T'$ to all the vertices $V$.
We will think of the MST $T$ as a tree rooted in $\rt$, and of $T'$ as a tree rooted at the fragment $F_1$ containing $\rt$. Using the information above, every vertex in each base fragment $F_i$ can learn the structure of the MST on the fragments $T'$, and infer its parent base fragment $p(F_i)$. If the MST edge connecting $F_i$ to $p(F_i)$ is $(u,v)$ (where $u\in F_i$), then $p(u)=v$, and we set $r_i=u$ to be the root of the base fragment $F_i$. For $F_1$, $r_1=\rt$ will be its root. Set $\mathcal{R}= \{r_1,r_2,\dots\}$ to be the set of base fragment root vertices. See \Cref{fig:LenghtLocalGlobal} for an illustration.

\begin{figure}[]
	\centering{\includegraphics[scale=0.7]{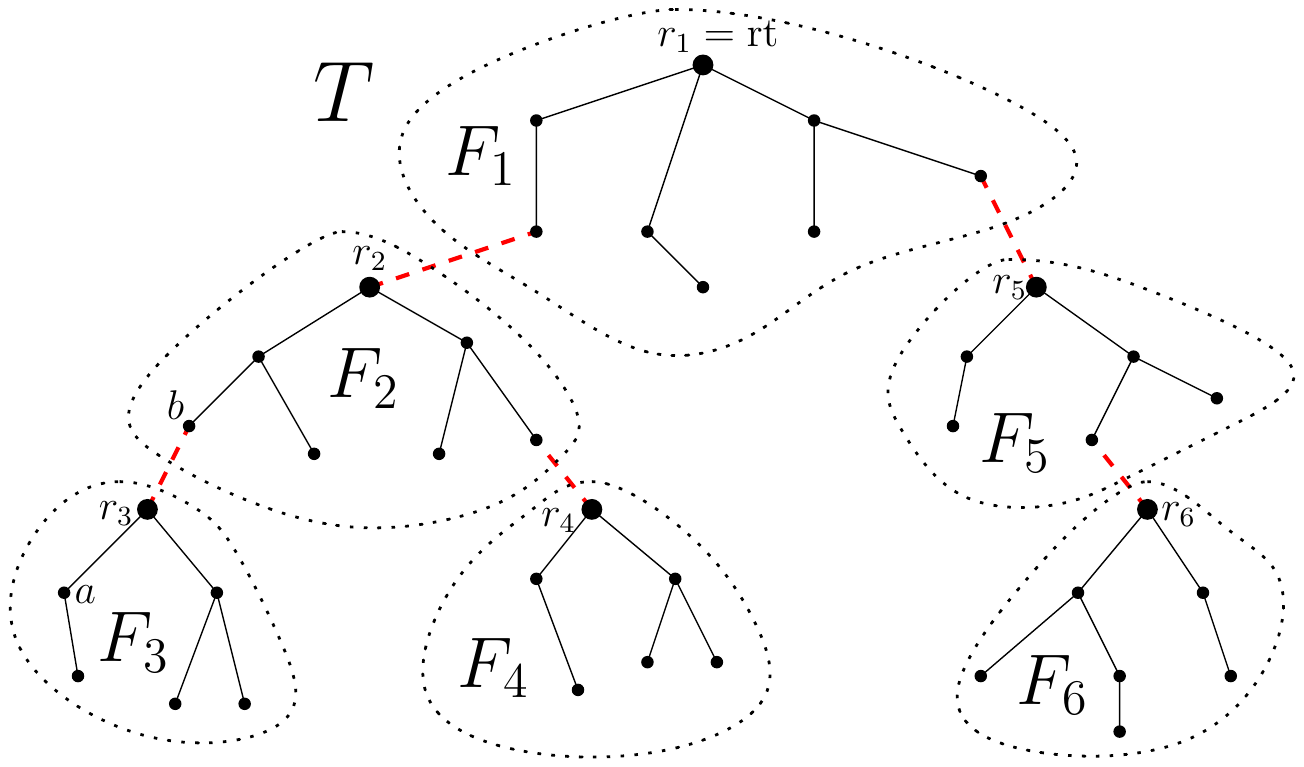}}
	\caption{\label{fig:LenghtLocalGlobal}\small \it
		The tree $T$ in the figure is divided to $6$ base fragments circled by a dotted line. The internal edges are colored black, while the external edges are dashed and colored red.
		The fragment $R_i$ is rooted in $r_i$, a vertex with outgoing edge towards a parent fragment. For example, consider the case where all the edges in $T$ have unit weight. Sample of local lengths: $\ell(a)=2$, $\ell(b)=0$, $\ell(r_1)=14$, and of global lengths $g(a)=2$, $g(b)=12$, $g(r_1)=78$.
	}
\end{figure}

\subsection{Computing Tour Lengths}
For $v\in F_i$, let $\ell(v)$ denote the length of the tour of the subtree of $F_i$ rooted at $v$. This is the local tour length for $v$, which is simply twice the sum of edge weights in that subtree. In addition, denote by $g(x)$ the length of the tour of the subtree of $T$ rooted at $v$, which is the global tour length for $v$. See \Cref{fig:LenghtLocalGlobal} for an example.

The computations of the local tour lengths is done locally in each base fragment, i.e. in all the base fragments in parallel. Consider $F_i$. Initially $\ell(v)=0$ for every leaf $v\in F_i$. Every intermediate vertex $u\in F_i$ that received messages from all its children in $F_i$, denoted $z_1,\dots,z_k$, computes
$\ell(u)=\sum_{i=1}^{k}\left(\ell(z_{i})+2\cdot w(u,z_{i})\right)$,
and sends $\ell(u)$ to its own parent.
Using the bounded \hd of the base fragments, this procedure will terminate in $O(\sqrt{n})$ rounds. When this stage concludes, every $v\in V$ knows $\ell(v)$.

After the computation of the local tour lengths, all the root vertices $\mathcal{R}$ broadcast to the entire graph their local tour lengths $\ell(r_1),\ell(r_2)\dots$ in $O(\sqrt{n}+D)$ rounds.
As the tree $T'$ is known to the entire graph, all the vertices can compute locally the global tour lengths $g(r_1),g(r_2)\dots$ of the roots. Specifically, for $r_i$ the root of the base fragment $F_i$, consider its descendent base fragments $\F'$ in $T'$. Denote by $r_F$ the root of the base fragment $F$ and by $e_F$ the external edge connecting $F$ to $p(F)$, its parent fragment. Then
\[
g(r_{i})=\ell(r_{i})+\sum_{F\in\F'}\left(\ell(r_{F})+2\cdot w(e_{F})\right)~.
\]
The computation of the global tour lengths for non-roots is done locally in a similar manner as the local tour lengths.
For a vertex $v\in F_i$, let $\tilde{z}_1,\dots,\tilde{z}_k$ be its children in $T$, then
\[
g(v)=\sum_{i=1}^{k}\left(g(\tilde{z}_{i})+2\cdot w(v,\tilde{z}_{i})\right)~.
\]
If $v$ is a leaf of some $F_i$ then it can compute $g(v)$ (since all its children are in ${\cal R}$), and then send $g(v)$ to it parent.
Every intermediate vertex $v$ of $F_i$ can compute $g(v)$ as soon as it received messages from all its children in $F_i$,
and then send $g(v)$ to its own parent. As we run this procedure in all the fragments in parallel, it terminates in $O(\sqrt{n})$ rounds. When this stage concludes, every $v\in V$ knows $g(v)$.

\subsection{Computing Tour Visit Times}

Finally we compute for every vertex $v\in V$ the set $\mL(v)$ and all its tour visiting times. This can be achieved by running a DFS search from the root $\rt$. Direct implementation of the DFS algorithm will take $\Omega(n)$ rounds. Instead, we will use a similar idea to the one above to speed-up the computation. First we will compute ``local''  DFS in the base fragments.
Then aggregate these times into a global DFS, first for the roots $\mathcal{R}$, and afterwards to all of $V$.

First we compute the ``local'' DFS visiting times. We will compute the visiting times in all the fragments in parallel. Consider $F_i$. For every $v\in F_i$ we will compute the entering and exit DFS time, for the global subtree of $T$ rooted at $r_i$ (rather than only in $F_i$). For a vertex $v\in F_i$, denote by $t_i(v)$ the interval between the first ``local'' entrance to the last exit.
The computation is performed top to bottom. First for $r_i$, $t_i(r_i)=[0,g(r_i)]$.
Next, consider a vertex $v\in F_i$ that received its interval $t_i(v)=[a,b]$. By induction $b-a=g(v)$. Denote by $z_1,\dots,z_k$ all the children of $v$ in $T$ (inside and outside $F_i$). Then $v$ will send each child $z_j$ its interval, where
\[
t_i(z_{j})=\left[a+\sum_{q<j}\left(g(z_{q})+2w(v,z_{q})\right)+w(v,z_{j})~,~a+\sum_{q<j}\left(g(z_{q})+2w(v,z_{q})\right)+w(v,z_{j})+g(z_{j})\right]~.
\]
Note that the length of the interval $t_i(z_j)$ is exactly $g(z_j)$. (We remark that roots in ${\cal R}$ do not initiate another interval assignment when they receive message from their parent in $T$.)
This procedure will terminate in $O(\sqrt{n})$ rounds (the \hd of any base fragment), where each vertex $v\in F_i$ knows  $t_i(v)$.
Moreover, consider a root vertex $r_i$ (other than $\rt$), such that $p(r_i)\in F_j$. Then by the description of the algorithm, in addition to $t_i(r_i)$, $t_i$ also knows $t_j(r_i)$, its interval in the DFS tour in the subtree rooted in $r_j$.

Finally we are ready to compute the global DFS intervals. Note that all we are actually missing here, is the first time visit of each root vertex. That is, we will compute a shift $s_i$ for every root $r_i$.
First, each root vertex $r_i$ broadcasts to $\rt$ (through $\tau$), its local interval $t_i(r_i)$ and its local interval in its parent fragment $t_j(r_i)$ (assuming $p(r_i)\in F_j$). This take  $O(\sqrt{n}+D)$ rounds.
Now, $\rt$ has all the required information in order to compute the DFS intervals for $\mathcal{R}$.  Initially $t(\rt)=t_1(\rt)$. Next, by induction assume that $\rt$ computed the interval $t(r_j)=[s_j,s_j+g(r_j)]$ for some $r_j\in\mathcal{R}$. Then for every $r_i\in \mathcal{R}$ such $p(r_i)\in F_j$ and $t_j(r_i)=[b,b+g(r_i)]$, we compute $t(r_i)=[s_j+b,s_j+b+g(r_i)]=[s_i,s_i+g(r_i)]$.
Eventually $\rt$ knows the global DFS intervals of all the roots $\mathcal{R}$ and can broadcast it to all the vertices in $O(\sqrt{n}+D)$ rounds.
Consider a vertex $v\in F_j$. Given its local interval $t_j(v)=[a,b]$ and the shift $s_j$, $v$ computes its global interval $t(v)=[s_j+a,s_j+b]$. This is done locally in all the fragments.

We conclude that in $\tilde{O}(\sqrt{n}+D)$ rounds every vertex $v$ can know its DFS interval. As every vertex can also compute the DFS intervals of it's children (in $T$), both $\mathcal{L}(v)$ and the visiting times of  $\mathcal{L}(v)$ are easily computed.

\section{Shallow Light Tree (SLT)}\label{sec:SLT}
In this section we present our SLT construction.
Recall that an $(\alpha,\beta)$-SLT of $G$ with a root $\rt$ is a tree $T_{\slt}$ that satisfies: 1) $\forall~ v\in V,~d_{T_{\slt}}(\rt,v)\le \alpha\cdot d_G(\rt,v)$, and 2) $w(T_{\slt})\le \beta\cdot w(MST)$. We show the following theorem.
\begin{theorem}[SLT]\label{thm:SLT}
	There is a deterministic distributed algorithm
	in the CONGEST model, that given a weighted graph $G$ with $n$ vertices and \hd $D$, root vertex $\rt$ and parameter $\eps>0$, constructs an $(1+\eps,1+O(\frac1\eps))$-SLT in  $\tilde{O}\left(\sqrt{n}+D\right)\cdot\poly(\eps^{-1})$ rounds.
\end{theorem}

Initially we will assume that $\eps\in(0,1)$. Afterwards, we will show how to generalize our result to $\eps\ge 1$.

Intuitively, in order to construct an SLT, one should combine the MST tree $T$ of $G$ with a shortest path tree rooted at $\rt$. Unfortunately, currently existing algorithms for constructing exact shortest path tree \cite{E17,GL18} require more than  $\tilde{O}(\sqrt{n}+D)$ rounds. Instead, we will use an approximate shortest path tree of \cite{BKKL17}. 
Specifically, they show that given a root vertex $\rt$ and a parameter $\eps\in (0,1]$, one can compute an approximate shortest path tree $T_{\rt}$ in  $\tilde{O}((\sqrt{n}+D)/ \poly(\eps))$ rounds. The approximation here is in the sense that for every vertex $v$,
\begin{equation}\label{eq:aspt}
d_{G}(\rt,v)\le d_{T_{\rt}}(\rt,v)\le (1+\eps)\cdot d_{G}(\rt,v)~.
\end{equation}

Our strategy to construct the SLT is similar to the framework of \cite{ABP92,KRY95}. First, construct an MST $T$ and an approximate shortest path $T_{\rt}$ (rooted at $\rt$).
Next, choose a subset of vertices $\bp$ called \emph{Break Points}. An intermediate graph $H$ will be constructed as a union of $T$, and the paths in $T_{\rt}$ from $\rt$ to all the vertices in $\bp$. We will argue that $H$ has lightness $O(\frac1\eps)$, and approximate distance to $\rt$ up to a $1+O(\eps)$ factor.
Our final SLT will be constructed as yet another approximate shortest path tree in $H$ (rooted at $\rt$).
The main difference from previous works is that a refined selection of breakpoints is required, in order to ensure efficient implementation in the CONGEST model.
In previous algorithms $\bp$ was chosen sequentially, i.e., break points were determined one after another. In contrast, we have two phases. In the first phase we choose the $\bp$ locally, while in the second phase we somewhat sparsify the set $\bp$ using a global computation.

We remark that \cite{KRY95} gave an efficient implementation of their algorithm in the PRAM CREW model with $n$ processors in $O(\log n)$ rounds. However, this implementation uses pointer jumping techniques which cannot be translated to the CONGEST model.

\subsection{Break Points Selection}\label{subsec:breakpoints}
Before picking the break points, we create traversal $\mL$ of the MST $T$ (rooted at $\rt$) as in \sectionref{sec:traversal}, such that each vertex $v\in V$ knows its appearances $\mL(v)$ and visiting times $R_x$ for any $x\in\mL(v)$. We will treat vertices with duplications according to their appearances on $\mL$. That is, $v$ will simulate different vertices in $\mL$ (and even can be chosen to $\bp$ several times). Note that every neighbor $u$ of $v$ in $G$ is a neighbor of exactly two vertices in $\mL(v)$ (w.r.t $\mL$), and therefore $v$ indeed can act in several roles without congestion issues.
In addition, each vertex $x_i\in\mL$ will know its index $i$ (i.e., how many vertices precede him in $\mL$ an not only the weighted visiting time $R_{x_i}$). This information can be obtained by running the same algorithm that finds visiting times, ignoring the weights.

Set $\alpha=\left\lceil \sqrt{n}\right\rceil$. We will construct $\bp$ in several steps.
The initial set of break points will be $\bp'=\{x_0,x_\alpha,x_{2\alpha},x_{3\alpha},\dots\}$, i.e.,  all the vertices whose index is a multiple of $\alpha$.
Next, we create a break point set $\bp_1$ from $\mL\setminus\bp'$. In each interval $I_i=\{x_{i\alpha},x_{i\alpha+1},\dots,x_{(i+1)\alpha-1}\}$ in parallel we will add points to $\bp_1$.
Initially, for every $i$, $x_{i\alpha}$ sends a message to $x_{i\alpha+1}$ with the information $\left(x_{i\alpha},R_{x_{i\alpha}}\right)$.
Generally, each vertex $x_j\in I_i$ will get at some point a message $\left(y,R_{y}\right)$ from $x_{j-1}$. The interpretation is that $y$ is the most recent addition to $\bp_1$, with the additional information of $R_{y}$.
Now, $x_j$ will join $\bp_1$ if the following condition holds:
\begin{equation}\label{eq:BPcondition}
d_{\mL}(x_j,y)=R_{x_j}-R_y>\eps\cdot d_{T_{\rt}}(\rt,x_j)~.
\end{equation}
Note that $d_{T_{\rt}}(\rt,x_j)$ is an information locally known to $x_j$ from the approximate shortest path computation.
Now, $x_j$ will send a message to $x_{j+1}$. If $x_j$ joined $\bp_1$, the message will be $\left(x_{j},R_{x_{j}}\right)$. Otherwise the message will be $\left(y,R_{y}\right)$.
After $\alpha-1$ rounds this procedure ends, and each internal vertex $x\notin \bp'$ knows whether it joins $\bp_1$ or not.

We cannot allow all the vertices in $\bp'$ to become break points (as we have no bound on the weight of all the shortest paths from $\rt$ towards them). Filtering which vertices among $\bp'$ will actually join $\bp$ will be done in a centralized fashion.
All the vertices $x_{i\alpha}\in \bp'$ will broadcast to $\rt$ the message $\left(x_{i\alpha},R_{x_{i\alpha}}\right)$.  By \lemmaref{lem:pipe} this can be done in $O(\sqrt{n}+D)$ rounds (since there are at most $2\sqrt{n}$ relevant indices $0<i\le 2n-2$).
The root $\rt$ locally creates the set $\bp_2\subseteq \bp'$
as follows. Initially $\rt=x_0$ joins $\bp_2$. Next, sequentially, for $x_{i\alpha}$ where $y\in \bp'$ was the last vertex to join $\bp_2$, $x_{i\alpha}$ will join $\bp_2$ if $d_{\mL}(x_{i\alpha},y)=R_{x_{i\alpha}}-R_y>\eps\cdot d_{T_{\rt}}(\rt,x_{i\alpha})$.
After this local computation, $\rt$ will broadcast to the entire graph the set $\bp_2$ in $O(\sqrt{n}+D)$ rounds. Define the final set of breakpoints as $\bp=\bp_1\cup\bp_2$.
(Intuitively, this step computes an SLT for the submetric of the original metric, induced by the set $\bp'$.)

\subsection{The Creation of $H$}
For a point $b\in\bp$, let  $P_b$ be the unique path from $\rt$ to $b$ in $T_{\rt}$. Let $H=T\cup\bigcup_{b\in\bp}P_b$. Denote by $A_{\bp}$ the set of vertices whose subtree in $T_{\rt}$ contains a vertex of $\bp$.
Each vertex knows whether it belongs to $\bp$, and we would like to ensure that every $v\in A_{\bp}$ will add an edge to its parent in $T_{\rt}$. Then adding the MST edges will conclude the construction of $H$. In the remaining part of this sub-section we show the computation of $A_{\bp}$.

We start by creating a set $\F$ of $O(\sqrt{n})$ fragments, which are subtrees of $T_{\rt}$ of \hd  $O(\sqrt{n})$ (this could be done by applying the first phase of the MST algorithm of \cite{KP98}, see \sectionref{sec:fragments} for more details). Each fragment can locally (in parallel) compute in $O(\sqrt{n})$ rounds whether it contains a break point, since it has bounded hop-diameter. Next, each fragment sends to $\rt$ its id, whether it contains a break point, and all of its outgoing edges in $T_{\rt}$. Note there is a total of $O(\sqrt{n})$ messages in this broadcast, so by \lemmaref{lem:pipe} this will take $O(\sqrt{n}+D)$ rounds (in fact, all vertices will receive all these messages). Now, $\rt$ can form a virtual tree $T'$ whose vertices are the fragments $\F=\{F_1,F_2,\dots\}$, and its edges connect fragments $F_i, F_j$ if there is an edge of $T_{\rt}$ between a vertex of $F_i$ to a vertex of $F_j$.
Now, we can also assign roots $r_1,r_2,\dots$ (where $r_i\in F_i$ and $r_1=\rt$), so that $r_i$ is the vertex with an edge in $T_{\rt}$ to a vertex in the parent of $F_i$ in $T'$. (See \sectionref{sec:traversal} for more details and a picture of the virtual tree and its roots.)

Then, $\rt$ is able to compute locally for every root $r_i\in F_i$ whether $r_i\in A_{\bp}$ (i.e.  its subtree contains a break point), simply by inspecting whether $F_i$ has a descendant in $T'$ with a break point. Then broadcast this information on all roots in $O(\sqrt{n}+D)$ rounds.
Note that now every local leaf $v\in F_i$ knows whether its subtree contains a break point.
Finally, locally in parallel in $O(\sqrt{n})$ rounds, in all the fragments $F_i$ we can compute for every vertex $v\in F_i$ whether $v\in A_{\bp}$.

\subsection{Stretch and Lightness Analysis}
In this subsection we argue that $H$ has the desired lightness and stretch.
We name the break points $\bp_1=\{b_0,b_1,\dots\}$, $\bp_2=\{\tilde{b}_0,\tilde{b}_1,\dots\}$ according to the order of their appearance in $\mL$.
It is clear by construction that for  every pair of consecutive break points $\tilde{b}_{j-1},\tilde{b}_j\in \bp_2$, $d_{\mL}(\tilde{b}_{j},\tilde{b}_{j-1})>\eps\cdot d_{T_{\rt}}(\rt,\tilde{b}_{j})$.
We claim that this property remain true also for every consecutive break points $b_{j-1},b_j\in\bp_1$. Indeed, let $i$ such that $b_j\in I_i$. If $b_{j-1}\in I_i$ then it follows by construction. Otherwise, $b_j$ is the first break point in $I_i$, and by \Cref{eq:BPcondition} it holds that $d_{\mL}(b_{j},b_{j-1})>d_{\mL}(b_{j},x_{i\alpha})>\eps\cdot d_{T_{\rt}}(\rt,b_{j})$.
\begin{corollary}\label{cor:SLTlight}
	$w(H)\le (1+\frac4\eps)\cdot w(T)$.
\end{corollary}
\begin{proof}
The graph $H$ consists of three parts, $w(H)\le w(T)+\sum_{b\in\bp_{1}}w(P_{b})+\sum_{\tilde{b}\in\bp_{2}}w(P_{\tilde{b}})$.
	We first bound the weight of the edges added due to $\bp_1$:
	\[
	\sum_{j\ge1}w(P_{b_{j}})=\sum_{j\ge1}d_{T_{\rt}}(\rt,b_{j})<\sum_{j\ge1}\frac{1}{\eps}\cdot d_{\mL}(b_{j-1},b_{j})\le\frac{1}{\eps}\cdot w(\mL)=\frac{2}{\eps}\cdot w(T)~.
	\]
	Similarly for $\bp_2$, $\sum_{j\ge1}w(P_{\tilde{b}_{j}})\le \frac{2}{\eps}\cdot w(T)$. The corollary follows.
\end{proof}
\begin{lemma}\label{lem:stretch-H}
	For every $v\in V$, $d_H(\rt,v)\le (1+25\eps)\cdot d_{G}(\rt,v)$.	
\end{lemma}
\begin{proof}
	Consider a vertex $v\in V$. Let $x$ be an arbitrary vertex from $\mL(v)$.
By construction there is a point $y\in\bp'\cup\bp_1$ such that
\begin{equation}\label{eq:y}
d_{\mL}(x,y)\le\eps\cdot d_{T_{\rt}}(\rt,x)~.
\end{equation}
	Moreover, for $y$ there is a point $y'\in \bp$ such that
\begin{equation}\label{eq:y'}
d_{\mL}(y,y')\le\eps\cdot d_{T_{\rt}}(\rt,y)~,
\end{equation}
(it might be that $x=y$ or $y=y'$).
	We first bound $d_{\mL}(x,y')$,	using that $d_G\le d_{\mL}$ and the assumption $\eps\le 1$.
	\begin{align*}
	d_{\mL}(x,y') & =d_{\mL}(x,y)+d_{\mL}(y,y')\\
	& \stackrel{\eqref{eq:y'}}{\le} d_{\mL}(x,y)+\epsilon\cdot d_{T_{\rt}}(\rt,y)\\
	& \stackrel{\eqref{eq:aspt}}{\le} d_{\mL}(x,y)+\epsilon\cdot (1+\eps)\cdot d_{G}(\rt,y)\\
    & \le  d_{\mL}(x,y)+2\eps\cdot (d_G(x,y)+d_{G}(\rt,x))\\
	& \le(1+2\epsilon)\cdot d_{\mL}(x,y)+2\eps\cdot d_{G}(\rt,x)\\
    & \stackrel{\eqref{eq:y}}{\le} (1+2\epsilon)\cdot\eps\cdot d_{T_{\rt}}(\rt,x)+2\eps\cdot d_{G}(\rt,x)\\
    & \stackrel{\eqref{eq:aspt}}{\le} (1+2\epsilon)\cdot\eps\cdot(1+\eps)\cdot d_{G}(\rt,x)+2\eps\cdot d_{G}(\rt,x)\\
    &\le 8\epsilon\cdot d_{G}(\rt,x)~.
	\end{align*}
	We conclude,	
	\begin{align*}
	d_{H}(\rt,x) & \le d_{T_{\rt}}(\rt,y')+d_{\mL}(x,y')\\
	& \le\left(1+\eps\right)\cdot\left(d_{G}(\rt,x)+d_{G}(x,y')\right)+d_{\mL}(x,y')\\
	& \le\left(1+\eps\right)\cdot d_{G}(\rt,x)+3\cdot d_{\mL}(x,y')\le\left(1+25\epsilon\right)\cdot d_{G}(\rt,x)~.
	\end{align*}
\end{proof}

\subsection{Finishing the Construction and Generalization to $\eps>1$}
After creating the subgraph $H$, we create a $(1+\eps)$-shortest path tree $T_{\slt}$ (using \cite{BKKL17}) of $H$ rooted at $\rt$ in $\tilde{O}(\sqrt{n}+D)/\poly(\eps)$ rounds. The tree $T_{\slt}$ has weight at most $w(T_{\slt})\le w(H)\le (1+\frac4\eps)\cdot w(T)$ by \Cref{cor:SLTlight}. Moreover, by \lemmaref{lem:stretch-H} for every vertex $v\in V$ it holds that
\[
d_{T_{\slt}}(\rt,v)\le (1+\eps)\cdot d_{H}(\rt,v)\le\left(1+\eps\right)\cdot(1+25\eps)\cdot d_{G}(\rt,v)\le(1+51\eps)\cdot d_{G}(\rt,v)~.
\]
By rescaling $\eps$, we conclude that for every $\eps\in (0,1)$ we can construct an $(1+\eps,O(\frac1\eps))$-SLT.
This is the right behavior (up to constant factors) when the distortion is small, as shown in \cite{KRY95}. We would like to obtain the inverse tradeoff, when the lightness is close to 1, say $1+\gamma$ for $0<\gamma<1$.
If we will directly apply our construction for large $1<\eps=1/\gamma$, then it can be checked that we will get distortion $O(1/\gamma^2)$ and lightness $1+\gamma$, instead of the desired $O(1/\gamma)$ distortion (roughly speaking, this is because a breakpoint in $BP'$ may have been removed, so in the analysis we applied a chain of two breakpoints).
Fortunately, we can use a reduction due to \cite{BFN16}.
\begin{lemma}[\cite{BFN16}]\label{lem:reduction_lightness}
	Let $G=(V,E)$ be a graph, $0<\delta<1$ a parameter and $t:{V\choose 2}\rightarrow\mathbb{R}_{+}$ some function. Suppose that we have an algorithm that for any given weight function $w:E\rightarrow\mathbb{R}_{+}$ constructs a spanner $H$ with lightness $\ell$ such that every pair $u,v\in V$ suffers distortion at most $t(u,v)$.
	Then for every weight function $w$ there exists a spanner $H$ with lightness $1+\delta\ell$  and such that every pair $u,v$ suffers distortion at most $t(u,v)/\delta$.
\end{lemma}
The reduction algorithm works by first changing the edge weights, and then executing the original algorithm. To compute the new weight of an edge $e\in E$, we only need to know the parameter $\delta$, the original weight $w(e)$ and whether $e$ belongs the MST. Thus we can easily use this reduction in the CONGEST model as well.

We presented an algorithm that constructs a subgraph $H$ with constant lightness (say $c$) and distortion $2$ from $\rt$. We will use distortion function below,
\[
t(u,v)=\begin{cases}
2 & \rt\in\left\{ u,v\right\} \\
\infty & \text{otherwise}
\end{cases}~.
\]
Thus, given any $0<\gamma<1$, we can apply \Cref{lem:reduction_lightness} with the parameter $\delta=\gamma/c$, and obtain lightness $1+\gamma$ and distortion $O(1/\gamma)$. \Cref{thm:SLT} now follows.

\section{Distributed Light Spanner}\label{sec:spanner}
In this section we devise an efficient distributed algorithm for light spanners in general graphs. In particular, we prove the following:
\begin{theorem}[Light Spanner]\label{thm:spanner}
	There is an randomized distributed algorithm
	in the CONGEST model, that given a weighted graph $G=(V,E,w)$ with $n$ vertices and \hd $D$, and parameters $k\in\N$, $\eps\in(0,1)$, in  $\tilde{O}_{\eps}\left(n^{\frac12+\frac{1}{4k+2}}+D\right)$ rounds, w.h.p. returns a $(2k-1)(1+\eps)$ spanner $H$ with $O_\eps(k\cdot n^{1+\frac{1}{k}})$ edges and lightness $O_\eps(k\cdot n^{\frac{1}{k}})$.
\end{theorem}

Our algorithm is similar in spirit to the algorithms of \cite{CDNS95,ES16,ENS15}. It
begins by computing a traversal $\mL=\{\rt=x_0,x_1,\dots,x_{2n-2}\}$ of the MST $T$, as in \Cref{sec:traversal}. In particular, every vertex $v$ knows the set of its appearances $\mL(v)$, and the visiting times and indices of every $x\in\mL(v)$.
Let $L=w(\mL)=2w(T)$ denote the length of $\mL$. Note that the value $L$ is known to all the vertices (or can be broadcasted in $O(D)$ rounds).

Set $E'=\{e\in E~:~ w(e)\le L/n\}$, and for every $i\in\{0,1,\dots,\left\lceil \log_{1+\eps}n\right\rceil\}$ set $E_i=\{e\in E~:~ \frac{L}{(1+\eps)^{i+1}}<w(e)\le \frac{L}{(1+\eps)^{i}}\}$. The algorithm constructs a different spanner for each edge set, and the final spanner will be a union of all these spanners.
First, build a spanner $H'$ for the low weight edges $E'$. This is done using the algorithm of Baswana and Sen \cite{BS07}. Specifically we run \cite{BS07} on the graph $G'=(V,E')$. In $O(k)$ rounds\footnote{The original paper claimed $O(k^2)$ rounds, but it has been observed that their algorithm can be implemented in $O(k)$ rounds.} we get a $(2k-1)$-spanner $H'$ of $G'$, where the expected number of edges is bounded by $O(k\cdot n^{1+1/k})$.

Next, for every $i\in\{0,1,\dots,\left\lceil \log_{1+\eps}n\right\rceil\}$ we will define a cluster graph ${\cal G}_i$, on which we will simulate a spanner for unweighted graphs. For each $i$, we partition $V$ into clusters ${\cal C}_i$.
Let $\mathcal{G}_i$ be an unweighted graph with $\mathcal{C}_i$ as its vertex set, and there is an edge between two clusters $A,B$ if there are vertices $a\in A,~b\in B$ such that $\{a,b\}\in E_i$.

In \cite{ES16,ENS15} the greedy spanner was applied on each ${\cal G}_i$. However, we cannot do so efficiently in a distributed setting. Instead, we will use the randomized algorithm of \cite{EN17spanner} on each ${\cal G}_i$. For an unweighted graph with $N$ vertices, that algorithm provides (with constant probability) a $(2k-1)$-spanner with $O(N^{1+1/k})$ edges, computed in $k$ rounds. Even though \cite{EN17spanner} gave an efficient distributed implementation, the input graph ${\cal G}_i$ is not the communication graph $G$. Our main technical contribution in this section is an adaptation of that algorithm for the cluster graphs ${\cal G}_i$, which also requires some changes in the partition that generates these graphs.

The algorithm of \cite{EN17spanner} runs in $k$ rounds. Initially, every vertex $x$ independently samples a value $r(x)$ from some distribution. In the first round $x$ initializes $m(x)=r(x)$, $s(x)=x$ and sends $(s(x),m(x)-1)$ to all its neighbors. In each following round, every vertex $x$ that received messages $\{(s(v),m(v))\}_{v\in N(x)}$ from its neighbors in the previous round, computes $u=\argmax_{v\in N^+(x)}\{m(v)\}$, updates $m(x)=m(u)$ and $s(x)=s(u)$, and sends $(s(x),m(x)-1)$ to all its neighbors. After $k$ rounds, each vertex $x$ adds to the spanner edges: for every vertex $y$, add one edge to an arbitrary vertex in the set $\{v\in N(x)~:~m(v)\ge m(x)-1\wedge s(v)=y\}$, if exists (in other words, for every source $y$ whose message reached $x$ with value at least $m(x)-1$, we add 1 edge to the spanner, from $x$ to a neighbor $v$ that sent $x$ the message on $y$). A useful property of the algorithm is that the stretch is guaranteed\footnote{For the stretch bound, the random samples $r(x)$ need to satisfy $r(x)<k$, which can be verified locally. The stretch analysis in \cite{EN17spanner} is conditioned on the event $\forall x\in V~:~r(x)<k$.} while the number of edges is bounded in expectation.

In order to implement this algorithm in ${\cal G}_i$, the vertices in each cluster $C\in {\cal C}_i$ need to compute the maximum over all the values they received from their neighbors in the previous round, and then send this value. Finally, we need to make sure that for every pair of clusters we want to connect, only one edge is added.
We will distinguish between two cases, as long as the hop diameter of clusters is not too large, they can compute locally the maximum value. When the hop diameter is too large, we will ensure that there are few clusters, and all the relevant information will be broadcasted to the entire graph.

\paragraph{Case 1: $i<\log_{1+\eps}(\eps\cdot n^{\frac{k}{2k+1}})$.}
Set $w_i=\frac{L}{(1+\eps)^i}$. We now describe the partition of $V$ into clusters $\mathcal{C}_i$. Each cluster $C\in{\cal C}_i$ has a name in $\{0,1,2,\dots,\frac{L}{\eps\cdot w_{i}}\}$, and its weak diameter is at most $\eps\cdot w_i$ w.r.t the MST metric (i.e. for any $u,v\in C$, $d_T(u,v)\le\eps\cdot w_i$). Let $v\in V$, and $x\in\mL(v)$ be an arbitrary appearance. Then $v$ will belong to the cluster $\left\lceil \frac{R_{x}}{\eps\cdot w_{i}}\right\rceil$
(recall that $R_x=d_{\cal L}(rt,x)$).
The weak diameter is indeed bounded, as for every $v,u\in V$ which both belong to the same cluster $j$, it holds that there are $x'\in\mL(v)$, $x''\in\mL(u)$ such that $\left|R_{x'}-R_{x''}\right|\le\eps\cdot w_{i}$, hence $d_{T}(u,v)\le d_{\mL}(x',x'')\le \eps\cdot w_i$.
Note that each vertex belongs to a single cluster. The number of clusters is bounded by $\left\lceil \frac{L}{\eps\cdot w_{i}}\right\rceil+1 =\left\lceil \frac{(1+\eps)^{i}}{\eps}\right\rceil+1 \le \left\lceil n^{\frac{k}{2k+1}}\right\rceil+1$.

Before the rounds simulations, each vertex $v\in V$ sends the identity of its cluster to all its neighbors. Additionally, $\rt$ samples a value $r_A$ for every cluster $A\in\mathcal{C}_i$, and broadcasts all these values to all the vertices in $O(|\mathcal{C}_i|+D)$ rounds, using \lemmaref{lem:pipe}.
Next, we describe how to implement a single round.
In the beginning of the round, each vertex knows the message $(s(A),m(A))$ that all clusters $A\in {\cal C}_i$ sent in the previous round. The simulation has three phases:
(1) Local phase: each vertex $v\in A$, computes the maximum $m(B)$ over all neighboring clusters $B$. No communication required, as $v$ knows the clusters of its neighbors and their messages.
(2) Convergecast phase: we convergecast $(s(A),m(A))$ towards $\rt$ on the BFS tree $\tau$. Each vertex $v$ that received all messages from its children in $\tau$ for a cluster $A$, will only forward the one with maximum $m(A)$. Therefore each vertex will forward only $|\mathcal{C}_i|$ messages, and we can pipeline all the messages of the second phase in $O(|\mathcal{C}_i|+D)$ rounds.
(3) Broadcast phase: the root $\rt$ broadcasts all the new messages $(s(A),m(A))$ for all clusters $A\in {\cal C}_i$ to all the graph in $O(|\mathcal{C}_i|+D)$ rounds.

After $k$ such rounds we add edges to the spanner by a
convergecast of all the spanner edges towards $\rt$ using $\tau$. Let $\mathcal{H}_i$ be the spanner of $\mathcal{G}_i$. In \cite{EN17spanner} it is shown that in expectation $|\mathcal{H}_i|=O(|\mathcal{C}_i|^{1+\frac1k})$.
Consider a vertex $v\in A$. For every cluster $B$ such that $\{A,B\}\in\mathcal{H}_i$ and there is a neighbor $u\in B$ of $v$, $v$ will send $((u,v),(A,B))$ towards $\rt$. On the other hand, each vertex receiving edges from  $A\times B$, will forward only a single such edge. After $O(|\mathcal{C}_i|^{1+\frac1k}+D)$ rounds the center $\rt$ knows $\mathcal{H}_i$, and for every edge $(A,B)\in\mathcal{H}_i$ it knows a representative $(a,b)\in A\times B$. In additional $O(|\mathcal{C}_i|^{1+\frac1k}+D)$ rounds $\rt$ broadcasts all these edges and $H_i$ is created accordingly. The total number of rounds to implement each iteration of \cite{EN17spanner} is
\[
\ensuremath{\ensuremath{O(|\mathcal{C}_{i}|^{1+\frac{1}{k}}+D)}\le\ensuremath{O\left(\left(n^{\frac{k}{2k+1}}\right)^{\frac{k+1}{k}}+D\right)=O\left(n^{\frac{1}{2}+\frac{1}{4k+2}}+D\right)}}~.
\]

\paragraph{Case 2: $\log_{1+\eps}(\eps\cdot n^{\frac{k}{2k+1}})<i\le \log_{1+\eps}(n)$.} Set $w_i=\frac{L}{(1+\eps)^i}$.
Similarly to the previous regime, we will partition the graph into clusters $\mathcal{C}_i$ with weak diameter $\eps\cdot w_i$. However, as the number of clusters will be large, computations will be done locally in the clusters. In order to make the local computations efficient, we will refine the partition into clusters such that each cluster will have bounded (weak) \hd.
We start by choosing cluster centers. A vertex $x_j\in\mL$ is a cluster center if one of the following conditions is fulfilled:
\begin{enumerate}
	\item There is an integer $s$ such that $R_{x_{j-1}}<s\cdot (\eps\cdot w_i)\le R_{x_{j}}$.
	\item $j$ is a multiple of $\left\lceil \frac{\eps\cdot n}{(1+\eps)^{i}}\right\rceil$ (that is, there is an integer $q$ such that $j=q\cdot\left\lceil \frac{\eps\cdot n}{(1+\eps)^{i}}\right\rceil$).
\end{enumerate}
Note that $x_0$ is a center. For every vertex $x_b\in\mL$, consider the closest center $x_a$ left of $x_b$ (w.r.t $\mL$). It holds that $R_{x_b}-R_{x_a}< \eps\cdot w_i$ and $b-a< \frac{\eps\cdot n}{(1+\eps)^{i}}$.
Moreover, the total number of centers is bounded by $\frac{L}{\eps\cdot w_{i}}+\frac{n}{\frac{\eps\cdot n}{(1+\eps)^{i}}}=\frac{2\cdot(1+\eps)^{i}}{\eps}$. In particular, each vertex can compute whether it is a center locally.
For every vertex $v\in V$, pick an arbitrary $x_j\in\mL(v)$, and let $j'\le j$ be the largest such that $x_{j'}$ is a center. Then $v$ joins the cluster $C(x_{j'})$ of $x_{j'}$.
If $x_a,x_b$ are two consecutive cluster centers, then $I(x_a)=\{x_a,x_{a+1},x_{a+2},\dots,x_{b-1}\}$ will be the communication interval for the cluster $C(x_a)$. Note that $C(x_a)\subseteq I(x_a)$ (they need not be equal, since each vertex $u\in V$ has several possible representatives in $\mL(u)$).

Note that the \hd of $I(x_a)$ is bounded by $\frac{\eps\cdot n}{(1+\eps)^{i}}\le n^{\frac{1}{2}+\frac{1}{4k+2}}$, and also for any $u,v\in C(x_a)$ we have $d_T(u,v)\le\eps\cdot w_i$.
Each vertex $v\in V$ belongs to a single cluster. However, $v$ might belong to many communication intervals. Nevertheless, every MST edge appears twice in $\mL$, and therefore it belongs to at most two communication intervals.

In order to enable the partition to clusters, each cluster center $x_a$ declares itself via $I(x_a)$. That is, it sends to the right neighbor (on $\mL$) a message declaring itself, which is forwarded until it reaches the next center $x_b$. This declaration takes $\frac{\eps\cdot n}{(1+\eps)^{i}}$ rounds. At the end, each vertex chooses to which cluster it joins, becomes aware of all the communication intervals it belongs to, and sends its cluster i.d. to all its neighbors.

Now that we have defined the clustering, the simulation of each iteration of \cite{EN17spanner} is done in essentially the same manner as the previous case, with the communication interval taking the role of the global BFS tree $\tau$. That is, in parallel for every cluster $C(x_a)$ we find the maximum over the $m(v)$ by convergecast in $I(x_a)$.
In the last round we convergecast the spanner edges touching the cluster $C(x_a)$, so we need a bound on that number. In \cite{EN17spanner} it is shown that w.h.p. every vertex (cluster) adds at most $O(|{\cal C}_i|^{1/k}\log n)=O(n^{1/k}\log n)$ edges to the spanner. So the number of rounds required for a simulation of a single iteration is at most
\[
O(n^{\frac1k}\log n+n^{\frac{1}{2}+\frac{1}{4k+2}})=O(n^{\frac{1}{2}+\frac{1}{4k+2}})~,
\]
(assuming $k>1$.)
The total number of rounds (for each $i$ in this range) is thus $O\left(k\cdot n^{\frac{1}{2}+\frac{1}{4k+2}}\right)$. This concludes the second case.

Our final spanner $H$ will be a union of the MST $T$, with the spanner $H'$ of $G'$, and with the spanners $H_i$ for all $0\le i\le \lceil\log_{1+\eps}n\rceil$.
As the there are $O(\log_{1+\eps}n)$ different scales, we conclude that the total construction of the spanner $H$ of $G$ takes $\tilde{O}(n^{\frac{1}{2}+\frac{1}{4k+2}}+D)$ rounds.

\subsection{Analysis}
In this section we finish the proof of \Cref{thm:spanner} by analyzing the  stretch, lightness, and sparsity of the spanner $H$.

\paragraph{Stretch.}
By the triangle inequality, it suffices to show that for every edge $\{u,v\}=e\in E$, it holds that $d_H(u,v)\le (2k-1)(1+\eps)w(e)$. In fact, we will show a bound of $(2k-1)(1+O(\eps))$ on the stretch. This can be fixed later by rescaling $\eps$.
Fix $\{u,v\}=e\in E$. We can assume that $w(e)\le L$, as otherwise we fulfill the requirement using the MST edges only. If $e\in E'$, then $d_H(u,v)\le d_{H'}(u,v)\le (2k-1)\cdot w(e)$.
Otherwise, let $i\ge 0$ such that $e\in E_i$, that is  $\frac{w_i}{(1+\eps)}<w(e)\le w_i$ for $w_i= \frac{L}{(1+\eps)^{i}}$.
Let $A_u,A_v\in\mathcal{C}_i$ be the clusters containing $u,v$ respectively. If $A_u=A_v$, then $d_H(u,v)\le d_{T}(u,v)\le \eps\cdot w_i\le w(e)$ (assuming $\eps<1/2$, say).
Otherwise, $\{A_u,A_v\}$ is an edge of $G_i$, and therefore there is a path  $A_u=A_0,A_1,\dots,A_t=A_v$  between $A_u,A_v$ in $\mathcal{H}_i$ where $t\le 2k-1$. In particular, for every $0\le j<t$, we added some edge $\{v_j,u_{j+1}\}\in A_j\times A_{j+1}\cap E_i$ to $H_i$. Let $u_0=u$ and $v_t=v$.
As the distance between every pair of vertices in any cluster is bounded by $\eps\cdot w_i$ and the weight of all the edges in $H_i$ is bounded by $w_i$ we conclude
\begin{align*}
d_{H}(u,v)\le d_{H_{i}\cup T}(u_{0},v_{t}) & \le d_{T}(u_{0},v_{0})+\sum_{j=0}^{t-1}\left(w(v_{j},u_{j+1})+d_{T}(u_{j+1},v_{j+1})\right)\\
& \le(t+1)\cdot\eps\cdot w_{i}+t\cdot w_{i}\\
& \le(2k-1)\cdot(1+O(\eps))\cdot w(e)~.
\end{align*}

\paragraph{Lightness.} We bound the lightness of $H'$ and each of the spanners $H_i$. First consider $H'$. Since the weight of every edge $e\in E'$ is at most $L/n$, we have $$w(H')\le |H'|\cdot \frac Ln=O(k\cdot n^{1+\frac1k}\cdot\frac Ln)=O(k\cdot n^{\frac1k}\cdot L)~.$$

Next consider $H_i$, which has expected $O(|\mathcal{C}_i|^{1+\frac1k})=O\left(\left(\frac{(1+\eps)^{i}}{\eps}\right)^{1+\frac{1}{k}}\right)$ edges, all of weight bounded by $w_i=\frac{L}{(1+\eps)^i}$. So the expected weight of all these $H_i$ together is
\begin{align*}
\sum_{i=0}^{\left\lceil \log_{1+\eps}n\right\rceil }\E\left[w(H_{i})\right] & \le\sum_{i=0}^{\left\lceil \log_{1+\eps}n\right\rceil }\E[|H_{i}|]\cdot w_{i}=\sum_{i=0}^{\left\lceil \log_{1+\eps}n\right\rceil }O\left(\left(\frac{(1+\eps)^{i}}{\eps}\right)^{1+\frac{1}{k}}\right)\cdot\frac{L}{(1+\eps)^{i}}\\
& =O\left(\frac{L}{\eps^{1+1/k}}\right)\cdot\sum_{i=0}^{\left\lceil \log_{1+\eps}n\right\rceil }(1+\eps)^{\frac{i}{k}}=O\left(\frac{L}{\eps^{1+1/k}}\right)\cdot\frac{(1+\eps)^{\frac{\left\lceil \log_{1+\eps}n\right\rceil+1}{k}}-1}{(1+\eps)^{1/k}-1}\\
&=O\left(\frac{L\cdot k\cdot n^{1/k}}{\eps^{2+1/k}}\right)~,
\end{align*}
where the last equality follows as
$(1+\eps)^{\frac{1}{k}}-1\ge e^{\frac{\eps}{2}\cdot\frac{1}{k}}-1\ge\frac{\eps}{2k}$.
We conclude that the expected weight of $H$ is
\[
\E[w(H)]\le w(T)+w(H')+\sum_{i=0}^{\left\lceil \log_{1+\eps}n\right\rceil }\E[w(H_{i})]=O_\eps\left(k\cdot n^{\frac{1}{k}}\cdot L\right)~.
\]

\paragraph{Sparsity.}
Following the analysis of the lightness, we have
\[
\sum_{i=0}^{\left\lceil \log_{1+\eps}n\right\rceil }\E[|H_{i}|]\le\sum_{i=0}^{\left\lceil \log_{1+\eps}n\right\rceil }O\left(\left(\frac{(1+\eps)^{i}}{\eps}\right)^{1+\frac{1}{k}}\right)=O\left(\frac{1}{\eps^{1+\frac1k}}\cdot\frac{n^{1+\frac{1}{k}}}{(1+\eps)^{1+\frac{1}{k}}-1}\right)=O\left(\frac{n^{1+\frac{1}{k}}}{\eps^{2+\frac1k}}\right)~,
\]
We conclude,
\[
\E[|H|]\le|T|+|H'|+\sum_{i=0}^{\left\lceil \log_{1+\eps}n\right\rceil }\E[|H_{i}|]=O_\eps\left(k\cdot n^{1+\frac{1}{k}}\right)~.
\]

\begin{remark}
We note that the number of edges mostly comes from the spanner $H'$. We can in fact use the techniques developed here in order to efficiently implement the algorithm of \cite{EN17spanner} for weighted graphs in the CONGEST model, which provides a $(2k-1)\cdot(1+\eps)$-spanner with $O_\eps(\log k\cdot n^{1+1/k})$ edges. That algorithm partitions the edges $E$ to $\approx\log k$ sets, and for each set, applies the unweighted version on a cluster graph. Since we already have an efficient distributed implementation of that unweighted algorithm, we conclude that our sparsity bound may be improved to $O_\eps(\log k\cdot n^{1+1/k})$. We leave the details to the full version.
\end{remark}

\paragraph{Successes Probability.}
Note that once the computation concludes, we can easily compute the size and lightness of the spanner in $O(D)$ rounds via the BFS tree $\tau$. Thus we can repeat the computation for $H'$ and each $H_i$ until they meet the required bounds, which will happen w.h.p. after at most $O(\log n)$ tries. Recall that the stretch bound is guaranteed to hold.

\section{Distributed Construction of Nets}\label{sec:nets}

In this section we devise an efficient distributed algorithm for computing nets in general graphs. Let $G=(V,E,w)$ be a weighted graph. Recall that for $\alpha>0$, a set $N\subseteq V$ is $\alpha$\emph{-covering } if for every vertex $x\in V$ there is $y\in N$ with $d_G(x,y)\le \alpha$.
A set $N\subseteq V$ is $\beta$\emph{-separated} if for every $x,y\in N$, $d_G(x,y)> \beta$.
We say that a set $N$ is an $(\alpha,\beta)$\emph{-net} if it is both $\alpha$-covering and $\beta$-separated.
(In the literature nets are often defined with $\alpha=\beta$, but we will need the more general definition, since we will only be able to provide nets with $\alpha>\beta$.)

Our construction of nets is inspired by the MIS (maximal independent set) algorithm of \cite{MRSZ11} (which itself is inspired by \cite{L86}).
The MIS algorithm works in $O(\log n)$ rounds, where in each round a permutation is sampled. A vertex joins the MIS iff it is local minimum (i.e., it appears before all its neighbors in the permutation). We will also sample a permutation. However, instead of checking only the neighbors, a vertex $v$ will join the net iff it is a local minimum in a geometric sense. I.e., $v$ appears before all the vertices of $B_G(v,\beta)$ in the permutation.

In order to implement this algorithm efficiently we will require several tools that have found distributed constructions recently, such as {\em Least-Element lists} and {\em shortest path trees}. However, we do not know how to compute these exactly in the allotted number of rounds, so we will settle for approximations.
The rest of the section is dedicated to proving the following theorem.
\begin{theorem}\label{lem:net}
	Given a weighted graph $G=(V,E,w)$ with \hd $D$ and parameters $\Delta>0$, $\delta\in(0,1)$, there is a randomized algorithm in the CONGEST model that computes  w.h.p. a $\left((1+\delta)\cdot\D,\frac{\D}{1+\delta}\right)$-net in $(\sqrt{n}+D)\cdot 2^{\tilde{O}(\sqrt{\log n}\cdot\log(1/\delta))}$ rounds.
\end{theorem}

\paragraph{Least Element Lists.}
LE lists were introduced by \cite{Coh97}:
\begin{definition}
	Given a weighted graph $G=(V,E,w)$, a set $A\subseteq V$ of vertices, and a permutation $\pi:A\rightarrow[|A|]$ on $A$, the LE list of a vertex $v\in A$ is defined as
	\[
	\text{LE}_{G,A,\pi}(v)=\left\{(u,d_G(u,v))~:~ u\in A,\nexists w\in A\text{ s.t. }d_{G}(v,w)\le d_{G}(v,u)\text{ and }\pi(w)<\pi(u)\right\} ~.
	\]
	In words, a vertex $u\in A$ joins $\text{LE}_{G,A,\pi}(v)$, the LE list of $v$, if $u$ is first in the permutation among all the vertices at distance at most $d_G(v,u)$ from $v$ (alternatively, $u$ is the closest vertex to $v$ among the first $\pi(u)$ vertices in the permutation.).
\end{definition}
Khan et. al. \cite{KKMPT12} showed that with high probability over the choice of the permutation $\pi$, it holds that  $|\text{LE}_{G,A,\pi}(v)|=O(\log |A|)$ simultaneously for all the vertices $v\in A$.
Using hopsets, Friedrichs and Lenzen \cite{FL16} were able to efficiently compute LE lists in the CONGEST model (improving upon Ghaffari and Lenzen \cite{GL14}) for a graph $H$ that is a good approximation of $G$. (We remark that their algorithm was given in the case $A=V$, but it is a simple adaptation to adjust it to the more general case.)
\begin{theorem}[\cite{FL16}]\label{thm:FL18}
	Consider a graph $G=(V,E)$ with $n$ vertices and \hd $D$, a set $A\subseteq V$, and let $\delta\in(0,1)$ be any parameter. There is a randomized algorithm in the CONGEST model that uniformly samples a permutation $\pi$ and computes $\{\text{LE}_{H,A,\pi}(v)\}_{v\in V}$ for a graph $H$ such that $d_G(u,v)\le d_H(u,v)\le (1+\delta)\cdot d_G(u,v)$. The algorithm is successful w.h.p., and the number of rounds is  $(\sqrt{n}+D)\cdot2^{\tilde{O}(\sqrt{\log n}\cdot\log(1/\delta))}$.
\end{theorem}

\paragraph{Algorithm.}
Here we describe the algorithm promised in \theoremref{lem:net}.
The algorithm will run in $O(\log n)$ iterations. Initially, set $A_1=V$ the set of active vertices, and $N=\emptyset$. We denote by $A_i$ the set of active vertices for the $i$'th iteration, and by $N_i$ the net just after the $i$'th iteration (so $N_0=\emptyset$).
	In the $i$'th iteration, apply \theoremref{thm:FL18} on the graph $G$ with the set $A_i$ and the parameter $\delta$. We obtain a (uniformly random) permutation $\pi_i$ on $A_i$, alongside with LE lists for all $v\in A_i$ w.r.t $\pi_i$ and the graph $H_i$ (which is a $1+\delta$ approximation of $G$).
	Every $v\in A_i$ will join $N_{i}$ iff it is the first in the permutation order among its $\D$-neighborhood (w.r.t $H_i$). In other words, $v$ joins $N_i$ if there is no $u\in\text{LE}_{H_i,A_i,\pi_i}(v)$, $u\neq v$, such that $d_{H_i}(u,v)\le\D$.
	
	Next, we will construct an $1+\delta$ approximate shortest path tree $T_{i}$ (using \cite{BKKL17}, say) rooted in $N_i$. Remove from $A_i$ every vertex at distance at most $(1+\delta)\cdot\D$ from $N_i$ (in $T_{i}$), to form $A_{i+1}$. This concludes a single iteration. Continue until iteration $i$ where $A_i=\emptyset$.
	
\paragraph{Running time.}
The number of rounds for computing the LE lists is $(\sqrt{n}+D)\cdot2^{\tilde{O}(\sqrt{\log n}\cdot\log(1/\delta))}$, and the shortest path tree takes $\tilde{O}((\sqrt{n}+D)/\poly(\delta))$ rounds. (Deciding whether to join $N_i$ is done locally once the LE lists are computed.)  We will soon show that w.h.p. there are $O(\log n)$ iterations, thus the total number of rounds is $(\sqrt{n}+D)\cdot2^{\tilde{O}(\sqrt{\log n}\cdot\log(1/\delta))}$.

\paragraph{Analysis.}
First we argue that once the algorithm concludes, $N$ is a $\left((1+\delta)\cdot\D,\frac{\D}{1+\delta}\right)$-net.
To see the packing property, consider $u,v\in N$, and we want to show that $d_G(u,v)>\frac{\D}{1+\delta}$.
If $u$ and $v$ joined $N$ in the same iteration $i$, then they both were maximal in their $\D$ neighborhoods w.r.t $H_i$, and therefore $d_G(u,v)\ge\frac{ d_{H_i}(u,v)}{1+\delta} >\frac{\D}{1+\delta} $.
Otherwise, assume w.l.o.g that $u$ joined $N$ at iteration $i$, while $v$ joined $N$ at iteration $i'>i$. As $v$ remained active, necessarily $(1+\delta)\cdot\D<d_{T_{i}}(v,N_i)\le(1+\delta)\cdot d_{G}(u,v)$.
For the covering property, let $v$ be some vertex that becomes inactive at iteration $i$. Then there is a vertex $u\in N$ such that $d_G(u,v)\le d_{T_{i}}(u,v)\le(1+\delta)\cdot\D$.

It remains to show that after $O(\log n)$ iterations w.h.p. no active vertices remain.
We say that a pair of vertices $\{u,v\}$ at distance at most $\D$ (w.r.t. $d_G$) is \emph{active} if both $u,v$ are active.
Set $\mathcal{E}_{0}=\left\{ \{u,v\}\mid d_{G}(u,v)\le\D\right\} $ to be the set of initially active pairs.  $\mathcal{E}_i$ will denote the set of pairs at distance at most $\D$ that remain active after the $i$'th iteration.
For an active vertex $v\in A_i$, denote by $A_i(v)=A_i\cap B_G(v,\D)$. Note that $u\in  A_i(v) \iff \{v,u\}\in\mathcal{E}_{i-1}$.
We say that $w\in A_i(v)$ \emph{kills} $v$ at the $i$'th iteration if $w$ is first in the permutation among $ A_i(v)\cup A_i(w)$. Suppose that $w$ kills $v$ in the $i$'th iteration. Then we claim that $w$ joins $N_i$, as for every vertex $w'\in A_i$ such that $d_{H_i}(w,w')\le \Delta$ it holds that $d_G(w,w')\le \Delta$ as well, and therefore $\pi(w)< \pi(w')$. In particular, $v$ will cease to be active as $d_{T_{i}}(v,N_i)\le (1+\delta)\cdot d_{G}(v,w)\le (1+\delta)\cdot\D$.
Note that at most one vertex $w$ can kill $v$ (as it must be the first in the permutation in $ A_i(v)$).
\footnote{Note also that it is possible for a vertex to cease to be active without being killed. This can happen if the first vertex in the permutation among $A_i(v)$ does not join $N_i$, while a different vertex in $A_i(v)$ does join. We do not take advantage of this possibility.} Therefore the probability that $v$ becomes inactive in iteration $i$ is at least $\sum_{w\in A_i}\Pr\left[v\text{ is killed by }w\right]$.
As we pick the permutation $\pi$ uniformly at random, the probability that $w$ kills $v$ is exactly $\left| A_i(v)\cup A_i(w)\right|^{-1}$.

Our next goal is to show that in expectation, the number of active pairs is halved in each iteration. Consider a pair $\{u,v\}\in\mathcal{E}_{i-1}$. Then
\begin{align*}\Pr\left[\left\{ u,v\right\} \notin\mathcal{E}_{i}\right] & \ge\frac{1}{2}\cdot\left(\Pr\left[u\text{ became inactive}\right]+\Pr\left[v\text{ became inactive}\right]\right)\\
& \ge\frac{1}{2}\cdot\sum_{w\in A_i}\Pr\left(\left[w\text{ kills }u\right]+\Pr\left[w\text{ kills }v\right]\right)~.
\end{align*}
$\mathcal{E}_{i-1}\setminus\mathcal{E}_{i}$ is the set of pairs who cease to be active in the $i$'th iteration.
We conclude	
\begin{align*}
\mathbb{E}\left[\left|\mathcal{E}_{i-1}\setminus\mathcal{E}_{i}\right|\right] & =\sum_{\left\{ u,v\right\} \in\mathcal{E}_{i-1}}\Pr\left[\left\{ u,v\right\} \notin\mathcal{E}_{i}\right]\\
& \ge\frac{1}{2}\cdot\sum_{\left\{ u,v\right\} \in\mathcal{E}_{i-1}}\sum_{w\in A_i}\Pr\left(\left[w\text{ kills }u\right]+\Pr\left[w\text{ kills }v\right]\right)\\
& =\frac{1}{2}\sum_{u\in A_i}\sum_{w\in A_i(u)}\Pr\left[w\text{ kills }u\right]\cdot\left| A_i(u)\right|\\
& =\frac{1}{2}\sum_{u\in A_i}\sum_{w\in A_i(u)}\frac{\left| A_i(u)\right|}{\left| A_i(u)\cup A_i(w)\right|}\\
& \ge\frac{1}{2}\sum_{u\in A_i}\sum_{w\in A_i(u)}\frac{\left| A_i(u)\right|}{\left| A_i(u)\right|+\left| A_i(w)\right|}\\
& =\frac{1}{2}\sum_{\left\{ u,w\right\} \in\mathcal{E}_{i-1}}\frac{\left| A_i(u)\right|+\left| A_i(w)\right|}{\left| A_i(u)\right|+\left| A_i(w)\right|}=\frac{1}{2}\cdot\left|\mathcal{E}_{i-1}\right|~.
\end{align*}

Set $p=\Pr\left[\left|\mathcal{E}_{i-1}\setminus\mathcal{E}_{i}\right|>\frac{1}{4}\cdot\left|\mathcal{E}_{i-1}\right|\right]$.
Then $\frac{1}{2}\cdot\left|\mathcal{E}_{i-1}\right|\le\mathbb{E}\left[\left|\mathcal{E}_{i-1}\setminus\mathcal{E}_{i}\right|\right]\le p\cdot\left|\mathcal{E}_{i-1}\right|+(1-p)\cdot\frac{1}{4}\cdot\left|\mathcal{E}_{i-1}\right|$,
which implies $p\ge\frac{1}{3}$.
After $O(\log n)$ iterations, by Chernoff inequality we had w.h.p. at least $\log_{\frac43} n$ iterations where $\left|\mathcal{E}_{i-1}\setminus\mathcal{E}_{i}\right|>\frac{1}{4}\cdot\left|\mathcal{E}_{i-1}\right|$, and therefore $\mathcal{E}_{O(\log n)}=\emptyset$. Note that if no active edges remain, then each active vertex is the only active vertex in its entire $\D$ neighborhood. In particular, in the next iteration all the active vertices will join $N$ and cease to be active, as required.

\section{Light Spanner for Doubling Metrics}\label{sec:doubling}

In this section we show a distributed algorithm that produces light spanners for doubling metrics.

\begin{theorem}[Light Spanner for Doubling Graphs]\label{thm:doubling}
	There is an randomized distributed algorithm
	in the CONGEST model, that given a weighted graph $G=(V,E,w)$ with $n$ vertices, \hd $D$, doubling dimension $\ddim$, and parameter $\eps\in(0,1)$, in  $(\sqrt{n}+D)\cdot\eps^{-\tilde{O}(\sqrt{\log n}+\ddim)}$ rounds, w.h.p. returns a $(1+\eps)$-spanner $H$ with $O\left(n\cdot \eps^{-O(\ddim)}\cdot\log n\right)$ edges and lightness $\eps^{-O(\ddim)}\cdot \log n$.
\end{theorem}

Our spanner construction will go as follows. We take an $\epsilon\D$-net for every distance scale $\D$, and connect net points that are within distance $\D$ of each other. An efficient construction of nets was described in \sectionref{sec:nets}. To efficiently implement the net point connections, we use approximate shortest path computations based on {\em hopsets} from every net point. This will ensure that every such approximate shortest path has few edges, thus the running time can be controlled.

The following lemma gives the standard packing property of doubling metrics (see, e.g.,  \cite{GKL03}).
\begin{lemma} \label{lem:doubling_packing}
	Let $(M,\rho)$ be a metric space  with doubling dimension $\ddim$.
	If $S \subseteq M$ is a subset of points with minimum interpoint distance $r$
	that is contained in a ball of radius $R$, then
	$|S| \le \left(\frac{2R}{r}\right)^{O(\ddim)}$.
\end{lemma}

\subsection{Spanner Construction}
Let $L$ be the weight of the MST of $G$.
For every $\D\in\{1,1+\eps,(1+\eps)^2,\dots,(1+\eps)^{\log_{1+\eps}L}\}$ we compute a $(2\eps\cdot\Delta,\frac{\eps}{2}\cdot\Delta)$-net as in \theoremref{lem:net} (e.g., we can take $\delta=1/2$). Let $N_i$ be the net for $\D=(1+\eps)^i$, and for every net point $u\in N_i$, run in parallel a $2\D$-bounded $(1+\eps)$-approximate shortest path tree rooted at $u$, and add to the spanner $H$ a $(1+\eps)$-approximate shortest path from $u$ to all other net points $v\in N_i$ found within this distance bound.

\paragraph{Shortest paths via hopsets.}
To implement these bounded shortest paths computations, we use the algorithm of \cite{EN16} based on hopsets. A $(\beta,\eps)$-hopset $F$ for a graph $G'=(V',E')$, is a set of edges $F$ that do not reduce distances, and for every $u,v\in V'$
\[
d_{G'\cup F}^{(\beta)}(u,v)\le (1+\eps)\cdot d_{G'}(u,v)~,
\]
where $d^{(\beta)}(\cdot,\cdot)$ is the length of the shortest path containing at most $\beta$ edges.
The graph $G'$ is created by choosing the set $V'\subseteq V$ of size $\approx\sqrt{n}\ln n$ at random, so that w.h.p. it intersects every shortest path in $G$ of length at least $\sqrt{n}$. The edges $E'$ are the $\sqrt{n}$-bounded distances in $G$ between the vertices of $V'$. Fix $\beta=1/\eps^{\tilde{O}(\sqrt{\log n})}$. In \cite{EN16} it is shown how to compute a $(\beta,\eps)$ hopset for $G'$ of size $O(\sqrt{n}\cdot\beta^2)$ in $O((\sqrt{n}+D)\cdot\beta^2)$ rounds. Furthermore, that hopset  is {\em path reporting}, that is, there is a path $P_e$ in $G$ for every hopset edge $e\in F$, such that the length of $P_e$ is exactly $w(e)$, and every vertex in $P_e$ knows it lies on a path implementing $e$.

Once a hopset $F$ is computed for $G'$, computing $(1+\eps)$-approximate shortest paths in $G$ from a root $v\in V$ amounts to running $\beta$ iterations of Bellman-Ford in $G\cup E'\cup F$ 
In order to perform a single Bellman-Ford iteration, we first send messages over the edges of $E$ for $2\sqrt{n}$ rounds (to reach a vertex of $V'$, and then discover the relevant edges of $E'$). Second, every $u\in V'$ broadcasts its distance estimate to the root $d(v)$ to all the graph using the BFS tree $\tau$ of depth $D$ (to implement the hopset edges). Since there are only $O(\sqrt{n}\cdot\ln n)$ vertices in $V'$, this can be done in $O(\sqrt{n}\cdot\ln n+D)$ rounds via \lemmaref{lem:pipe}.

In our setting we would like to compute in parallel multiple source $2\D$-bounded approximate shortest paths. To this end, we will use the fact the shortest path metric is doubling, hence every vertex will participate in a small number of such computations.

\subsection{Analysis}

\paragraph{Running Time.}
Fix some $\D=(1+\eps)^i$. Computing the $(2\eps\cdot\Delta,\frac{\eps}{2}\cdot\Delta)$-net $N_i$ takes $(\sqrt{n}+D)\cdot2^{\tilde{O}(\sqrt{\log n})}$ rounds. We next analyze the running time of the shortest paths computation from all net points in $N_i$, which is conducted in parallel.
By \lemmaref{lem:doubling_packing} we have that the number of net points $N_i$ in any ball of radius $2\D$ is $\eps^{-O(\ddim)}$. This suggests that for any $v\in V$ and any hopset edge $(x,y)\in F$, there are at most $\eps^{-O(\ddim)}$ sources of shortest paths computations that will reach them.
Thus the number of rounds required to implement in parallel all approximate shortest path computations is $O\left((\sqrt{n}+D)\cdot\beta\cdot\eps^{-O(\ddim)}\right)=(\sqrt{n}+D)\cdot\eps^{-\tilde{O}(\sqrt{\log n}+\ddim)}$. This is proportional to the time required to compute the hopset, and as there are $O(\log L)=O(\log n)$ different scales, the final running time is $(\sqrt{n}+D)\cdot\eps^{-\tilde{O}(\sqrt{\log n}+\ddim)}$.

\paragraph{Stretch Bound.}	
For simplicity, we will prove stretch $1+c\cdot \eps$ for some constant $c$ to be determined later. One can get stretch $1+\eps$ by rescaling $\eps$.
Consider a pair $u,v\in V$ such that $(1+\eps)^{i-1}<d_{G}(u,v)\le(1+\eps)^{i}=\D$.
Assume by induction that every pair $u',v'$ at distance at most $(1+\eps)^{i-1}$ already enjoys stretch at most $1+c\cdot\eps$ in $H$. The base case $i=0$ is trivial since there are no pairs $u'\neq v'$ of distance less than 1.
Let $\tilde{u},\tilde{v}\in N_i$ be net points such that $d_{G}(u,\tilde{u}),d_{G}(v,\tilde{v})\le2\eps\cdot \Delta$.
By the triangle inequality $d_{G}(\tilde{u},\tilde{v})\le (1+4\eps)\cdot\Delta$, and since $\eps<1/8$ we have that $(1+\eps)\cdot(1+4\eps)\cdot\D\le 2\D$, so the $2\D$-bounded shortest path exploration from $\tilde{u}$ must have discovered $\tilde{v}$.
Therefore we added a $1+\eps$ approximate shortest path between $\tilde{u}$ and $\tilde{v}$ to $H$.  Using the induction hypothesis on the pairs $\{u,\tilde{u}\}$ and $\{v,\tilde{v}\}$, we conclude
\begin{align*}
d_{H}\left(v,u\right) & \le d_{H}\left(v,\tilde{v}\right)+d_{H}\left(\tilde{v},\tilde{u}\right)+d_{H}\left(u,\tilde{u}\right)\\
& \le(1+c\eps)\cdot2\epsilon\Delta+(1+\eps)(1+4\epsilon)\Delta+(1+c\eps)\cdot2\epsilon\Delta\\
& \overset{(*)}{<}\frac{1+c\eps}{1+\epsilon}\cdot\Delta\le(1+c\eps)\cdot d_{G}\left(u,v\right)~,
\end{align*}
where the inequality $(*)$ follows for any constant $c\ge 30$, using that $\eps<1/8$.

\paragraph{Lightness bound.}
Let $n_i=|N_i|$.
In \cite{FN18} the following claim was shown.
\begin{claim}[\cite{FN18}]\label{clm:FN18}
	Consider a weighted graph $G$ with MST of weight $L$, such that there is an $r$-separated set $N$. Then $|N|\le\left\lceil\frac{2L}{r}\right\rceil $.
\end{claim}
It follows that $n_i=O(\frac{L}{\eps\Delta})$.
Recall that \lemmaref{lem:doubling_packing} implies that the number of net points $N_i$ in a ball of radius $2\D$ is $\eps^{-O(\ddim)}$. So for every net point $u\in N_i$ we added at most $\eps^{-O(\ddim)}$ paths of weight at most $2\Delta$ each. Thus the total weight of edges added for the $i$'th scale is bounded by $n_i\cdot \eps^{-O(\ddim)}\cdot 2\Delta=\eps^{-O(\ddim)}\cdot L$.
In particular, the sum of weights of the edges added in the $2\log_{1+\eps}n$ scales $i\in\{\log_{1+\eps}\frac {L}{n^2},\dots, \log_{1+\eps}L\}$ is bounded by $\eps^{-O(\ddim)}\cdot\log n\cdot L$. The contribution of all the other scales is negligible as all these scales adds at most $n^2$ edges of weight less than $\frac{L}{n^2}$. The bound on the lightness follows.

\paragraph{Sparsity bound.}
Consider a vertex $v\in V$ and the set of edges added for scale $\D=(1+\eps)^i$. Consider the ball $B$ of radius $2\D$ around $v$, recall that $|B\cap N_i|=\eps^{-O(\ddim)}$. We added a path to the spanner only between net points of distance at most $2\D$. Therefore $v$ can participate in such paths only when both net points are in $B$ (as otherwise the length of a path going through $v$ will be greater than $2\D$). Therefore $v$ might participate in at most $(\eps^{-O(\ddim)})^2=\eps^{-O(\ddim)}$ such paths. As we added at most $2$ edges touching $v$ per path, the number of edges added which are incident on $v$ is bounded by $\eps^{-O(\ddim)}$.
As there are $\log_{1+\eps}L$ scales (and $n$ vertices), we can bound the total number of edges added to the spanner by $n\cdot \eps^{-O(\ddim)}\cdot\log n$ (recall we assumed $G$ has diameter $\poly(n)$, thus also $L=\poly(n)$).

\section{Lower Bounds}\label{sec:lower}

In a seminal paper, Das Sarma et al. \cite{SHKKNPPW12} showed that approximating the weight of an MST up to polynomial factors takes $\tilde{\Omega}(\sqrt{n})$ rounds. As both an SLT and a light spanner provide such an approximation to the weight of the MST, we conclude:

\begin{theorem}\label{thm:LB}
	Every distributed algorithm in the CONGEST model that computes an SLT or a spanner with polynomial lightness, takes $\tilde{\Omega}(\sqrt{n}+D)$ rounds.
\end{theorem}

Next we argue that computing nets takes $\tilde{\Omega}(\sqrt{n}+D)$ rounds as well. Our lower bound is for general graphs. Therefore, it is possible that computing nets for graphs with bounded doubling dimension can be performed faster. Nevertheless, we conclude that \Cref{lem:net} is tight (up to lower order terms).

\begin{theorem}\label{thm:netLB}
	Suppose that there is a distributed algorithm in the CONGEST model that for every parameter $\D$ computes an $(\alpha\cdot\D,\D)$-net~ for some $1\le\alpha\le\poly(n)$. Then the algorithm runs in $\tilde{\Omega}(\sqrt{n}+D)$ rounds.
\end{theorem}
\begin{proof}
	Let $G$ be a graph from the family \cite{SHKKNPPW12} used for their lower bound for approximating the weight of the MST. The only property we will use is that $G$ has polynomial diameter $\Lambda$. (W.l.o.g. the minimal distance is 1, and $\Lambda$ is the largest distance in $G$.)
	We will create nets in all the distance scales, and use their cardinality in order to provide a polynomial approximation to the weight of the MST.
	As there are only $O(\log n)$ distance scales, and the cardinality of all the nets can be computed in $O(D+ \log n)$ time, the lower will follow.
	
	For every $i\ge-\lceil\log \alpha\rceil$, compute an $(\alpha\cdot 2^i, 2^i)$-net~ $N_i$. We stop in the first time that $|N_i|=1$. Note that we compute at most $O(\log n)$ nets, since $\log\alpha=O(\log n)$, and when $2^i$ is larger than $\Lambda=\poly(n)$, there can be only a single net point.
	Next, we compute the cardinality of all the net points and return $\Psi=\sum_{i}n_{i}\cdot\alpha\cdot2^{i+1}$ where $n_i=|N_i|$. To finish the proof, we will argue that $L\le\Psi\le O(\alpha\cdot\log n)\cdot L$ (where $L$ is the weight of the MST).

	For every $i$, as $N_i$ is $2^i$-separated, by \Cref{clm:FN18} $n_i\le \left\lceil \frac{2L}{2^{i}}\right\rceil $.
	In particular, $\Psi=\sum_{i}n_{i}\cdot\alpha\cdot2^{i+1}\le\alpha\cdot\sum_{i}\left\lceil \frac{2L}{2^{i}}\right\rceil\cdot2^{i+1}=O(\alpha\cdot\log n)\cdot L$.
	
	For the second inequality, we will construct a connected subgraph $H$ of $G$. For every net point $x\in N_i$, let $y\in N_{i+1}$ be the closest point to $x$ among points in $N_{i+1}$. Then by the covering property $d_G(x,y)\le \alpha\cdot 2^{i+1}$. Add to $H$ the shortest path from $x$ to $y$. Do this for all the nets. Note that $N_{-\lceil\log \alpha\rceil}=V$ (since any point can cover only itself), and that the maximal net consist of a single point. Therefore $H$ is connected.
	We conclude $L\le w(H)\le\sum_{i}n_{i}\cdot\alpha\cdot2^{i+1}=\Psi$.		
\end{proof}

\bibliographystyle{alphaurlinit}
\bibliography{DistSpan}

\newcommand{\etalchar}[1]{$^{#1}$}
\begin{thebibliography}{KKM{\etalchar{+}}12}

\bibitem[AB16]{AB16}
A.~Abboud and G.~Bodwin.
\newblock The 4/3 additive spanner exponent is tight.
\newblock In {\em Proceedings of the 48th Annual {ACM} {SIGACT} Symposium on
  Theory of Computing, {STOC} 2016, Cambridge, MA, USA, June 18-21, 2016},
  pages 351--361,
\newblock 2016, \href {http://dx.doi.org/10.1145/2897518.2897555}
  {\path{doi:10.1145/2897518.2897555}}.

\bibitem[ABP90]{ABP90}
B.~Awerbuch, A.~E. Baratz, and D.~Peleg.
\newblock Cost-sensitive analysis of communication protocols.
\newblock In {\em Proceedings of the Ninth Annual {ACM} Symposium on Principles
  of Distributed Computing, Quebec City, Quebec, Canada, August 22-24, 1990},
  pages 177--187,
\newblock 1990, \href {http://dx.doi.org/10.1145/93385.93417}
  {\path{doi:10.1145/93385.93417}}.

\bibitem[ABP92]{ABP92}
B.~Awerbuch, A.~E. Baratz, and D.~Peleg.
\newblock Efficient broadcast and light-weight spanners.
\newblock Technical Report CS92-22, The Weizmann Institute of Science, Rehovot,
  Israel.,
\newblock 1992.

\bibitem[ACIM99]{ACIM99}
D.~Aingworth, C.~Chekuri, P.~Indyk, and R.~Motwani.
\newblock Fast estimation of diameter and shortest paths (without matrix
  multiplication).
\newblock {\em {SIAM} J. Comput.}, 28(4):1167--1181,
\newblock 1999, \href {http://dx.doi.org/10.1137/S0097539796303421}
  {\path{doi:10.1137/S0097539796303421}}.

\bibitem[ADD{\etalchar{+}}93]{ADDJS93}
I.~Alth{\"{o}}fer, G.~Das, D.~P. Dobkin, D.~Joseph, and J.~Soares.
\newblock On sparse spanners of weighted graphs.
\newblock {\em Discrete {\&} Computational Geometry}, 9:81--100,
\newblock 1993, \href {http://dx.doi.org/10.1007/BF02189308}
  {\path{doi:10.1007/BF02189308}}.

\bibitem[ADF{\etalchar{+}}17]{ADFSW17}
S.~Alstrup, S.~Dahlgaard, A.~Filtser, M.~St{\"{o}}ckel, and C.~Wulff{-}Nilsen.
\newblock Constructing light spanners deterministically in near-linear time.
\newblock {\em CoRR}, abs/1709.01960,
\newblock 2017, \href {http://arxiv.org/abs/1709.01960}
  {\path{arXiv:1709.01960}}.

\bibitem[AGLP89]{AGLP89}
B.~Awerbuch, A.~V. Goldberg, M.~Luby, and S.~A. Plotkin.
\newblock Network decomposition and locality in distributed computation.
\newblock In {\em 30th Annual Symposium on Foundations of Computer Science,
  Research Triangle Park, North Carolina, USA, 30 October - 1 November 1989},
  pages 364--369,
\newblock 1989, \href {http://dx.doi.org/10.1109/SFCS.1989.63504}
  {\path{doi:10.1109/SFCS.1989.63504}}.

\bibitem[Ass83]{Ass83}
P.~Assouad.
\newblock Plongements lipschitziens dans $\mathbb{R}^n$.
\newblock {\em Bull. Soc. Math. France}, 111(4):429--448, 1983.
\newblock
\newblock \url{http://eudml.org/doc/87452}.

\bibitem[Awe85]{A85}
B.~Awerbuch.
\newblock Complexity of network synchronization.
\newblock {\em J. ACM}, 32(4):804--823,
\newblock October 1985, \href {http://dx.doi.org/10.1145/4221.4227}
  {\path{doi:10.1145/4221.4227}}.

\bibitem[BDS04]{SDS04}
Y.~Ben{-}Shimol, A.~Dvir, and M.~Segal.
\newblock {SPLAST:} a novel approach for multicasting in mobile wireless ad hoc
  networks.
\newblock In {\em Proceedings of the {IEEE} 15th International Symposium on
  Personal, Indoor and Mobile Radio Communications, {PIMRC} 2004, 5-8 September
  2004, Barcelona, Spain}, pages 1011--1015,
\newblock 2004, \href {http://dx.doi.org/10.1109/PIMRC.2004.1373851}
  {\path{doi:10.1109/PIMRC.2004.1373851}}.

\bibitem[BEPS12]{BEPS12}
L.~Barenboim, M.~Elkin, S.~Pettie, and J.~Schneider.
\newblock The locality of distributed symmetry breaking.
\newblock In {\em 53rd Annual {IEEE} Symposium on Foundations of Computer
  Science, {FOCS} 2012, New Brunswick, NJ, USA, October 20-23, 2012}, pages
  321--330,
\newblock 2012, \href {http://dx.doi.org/10.1109/FOCS.2012.60}
  {\path{doi:10.1109/FOCS.2012.60}}.

\bibitem[BFN16]{BFN16}
Y.~Bartal, A.~Filtser, and O.~Neiman.
\newblock On notions of distortion and an almost minimum spanning tree with
  constant average distortion.
\newblock In {\em SODA 2016}, pages 873--882, 2016.
\newblock
\newblock The full version of this paper (containing the mentioned reduction)
  appears in \url{http://arxiv.org/abs/1609.08801}.

\bibitem[BKKL17]{BKKL17}
R.~Becker, A.~Karrenbauer, S.~Krinninger, and C.~Lenzen.
\newblock Near-optimal approximate shortest paths and transshipment in
  distributed and streaming models.
\newblock In {\em 31st International Symposium on Distributed Computing, {DISC}
  2017, October 16-20, 2017, Vienna, Austria}, pages 7:1--7:16,
\newblock 2017, \href {http://dx.doi.org/10.4230/LIPIcs.DISC.2017.7}
  {\path{doi:10.4230/LIPIcs.DISC.2017.7}}.

\bibitem[BLW17]{BLW17}
G.~Borradaile, H.~Le, and C.~Wulff{-}Nilsen.
\newblock Minor-free graphs have light spanners.
\newblock In {\em 58th {IEEE} Annual Symposium on Foundations of Computer
  Science, {FOCS} 2017, Berkeley, CA, USA, October 15-17, 2017}, pages
  767--778,
\newblock 2017, \href {http://dx.doi.org/10.1109/FOCS.2017.76}
  {\path{doi:10.1109/FOCS.2017.76}}.

\bibitem[BLW19]{BLW19}
G.~Borradaile, H.~Le, and C.~Wulff{-}Nilsen.
\newblock Greedy spanners are optimal in doubling metrics.
\newblock In {\em Proceedings of the Thirtieth Annual {ACM-SIAM} Symposium on
  Discrete Algorithms, {SODA} 2019, San Diego, California, USA, January 6-9,
  2019}, pages 2371--2379,
\newblock 2019, \href {http://dx.doi.org/10.1137/1.9781611975482.145}
  {\path{doi:10.1137/1.9781611975482.145}}.

\bibitem[BS07]{BS07}
S.~Baswana and S.~Sen.
\newblock A simple and linear time randomized algorithm for computing sparse
  spanners in weighted graphs.
\newblock {\em Random Struct. Algorithms}, 30(4):532--563,
\newblock 2007, \href {http://dx.doi.org/10.1002/rsa.20130}
  {\path{doi:10.1002/rsa.20130}}.

\bibitem[CDNS95]{CDNS95}
B.~Chandra, G.~Das, G.~Narasimhan, and J.~Soares.
\newblock New sparseness results on graph spanners.
\newblock {\em Int. J. Comput. Geometry Appl.}, 5:125--144,
\newblock 1995, \href {http://dx.doi.org/10.1142/S0218195995000088}
  {\path{doi:10.1142/S0218195995000088}}.

\bibitem[Coh97]{Coh97}
E.~Cohen.
\newblock Size-estimation framework with applications to transitive closure and
  reachability.
\newblock {\em J. Comput. Syst. Sci.}, 55(3):441--453,
\newblock 1997, \href {http://dx.doi.org/10.1006/jcss.1997.1534}
  {\path{doi:10.1006/jcss.1997.1534}}.

\bibitem[Coh98]{Coh98}
E.~Cohen.
\newblock Fast algorithms for constructing t-spanners and paths with stretch t.
\newblock {\em {SIAM} J. Comput.}, 28(1):210--236,
\newblock 1998, \href {http://dx.doi.org/10.1137/S0097539794261295}
  {\path{doi:10.1137/S0097539794261295}}.

\bibitem[CW18]{CW18}
S.~Chechik and C.~Wulff{-}Nilsen.
\newblock Near-optimal light spanners.
\newblock {\em {ACM} Trans. Algorithms}, 14(3):33:1--33:15,
\newblock 2018, \href {http://dx.doi.org/10.1145/3199607}
  {\path{doi:10.1145/3199607}}.

\bibitem[DGPV08]{DGPV08}
B.~Derbel, C.~Gavoille, D.~Peleg, and L.~Viennot.
\newblock On the locality of distributed sparse spanner construction.
\newblock In {\em Proceedings of the Twenty-Seventh Annual {ACM} Symposium on
  Principles of Distributed Computing, {PODC} 2008, Toronto, Canada, August
  18-21, 2008}, pages 273--282,
\newblock 2008, \href {http://dx.doi.org/10.1145/1400751.1400788}
  {\path{doi:10.1145/1400751.1400788}}.

\bibitem[DHN93]{DHN93}
G.~Das, P.~J. Heffernan, and G.~Narasimhan.
\newblock Optimally sparse spanners in 3-dimensional euclidean space.
\newblock In {\em Proceedings of the Ninth Annual Symposium on Computational
  GeometrySan Diego, CA, USA, May 19-21, 1993}, pages 53--62,
\newblock 1993, \href {http://dx.doi.org/10.1145/160985.160998}
  {\path{doi:10.1145/160985.160998}}.

\bibitem[DPP06]{DPP06}
M.~Damian, S.~Pandit, and S.~Pemmaraju.
\newblock Distributed spanner construction in doubling metric spaces.
\newblock In {\em Proceedings of the 10th International Conference on
  Principles of Distributed Systems}, OPODIS'06, pages 157--171, Berlin,
  Heidelberg, 2006.
\newblock Springer-Verlag, \href {http://dx.doi.org/10.1007/11945529_12}
  {\path{doi:10.1007/11945529_12}}.

\bibitem[Elk04]{E04}
M.~Elkin.
\newblock Unconditional lower bounds on the time-approximation tradeoffs for
  the distributed minimum spanning tree problem.
\newblock In {\em Proceedings of the 36th Annual {ACM} Symposium on Theory of
  Computing, Chicago, IL, USA, June 13-16, 2004}, pages 331--340,
\newblock 2004, \href {http://dx.doi.org/10.1145/1007352.1007407}
  {\path{doi:10.1145/1007352.1007407}}.

\bibitem[Elk07]{E06}
M.~Elkin.
\newblock A near-optimal distributed fully dynamic algorithm for maintaining
  sparse spanners.
\newblock In {\em Proceedings of the Twenty-Sixth Annual {ACM} Symposium on
  Principles of Distributed Computing, {PODC} 2007, Portland, Oregon, USA,
  August 12-15, 2007}, pages 185--194,
\newblock 2007, \href {http://dx.doi.org/10.1145/1281100.1281128}
  {\path{doi:10.1145/1281100.1281128}}.

\bibitem[Elk17a]{E17}
M.~Elkin.
\newblock Distributed exact shortest paths in sublinear time.
\newblock In {\em Proceedings of the 49th Annual {ACM} {SIGACT} Symposium on
  Theory of Computing, {STOC} 2017, Montreal, QC, Canada, June 19-23, 2017},
  pages 757--770,
\newblock 2017, \href {http://dx.doi.org/10.1145/3055399.3055452}
  {\path{doi:10.1145/3055399.3055452}}.

\bibitem[Elk17b]{Elk17}
M.~Elkin.
\newblock A simple deterministic distributed {MST} algorithm, with near-optimal
  time and message complexities.
\newblock In {\em Proceedings of the {ACM} Symposium on Principles of
  Distributed Computing, {PODC} 2017, Washington, DC, USA, July 25-27, 2017},
  pages 157--163,
\newblock 2017, \href {http://dx.doi.org/10.1145/3087801.3087823}
  {\path{doi:10.1145/3087801.3087823}}.

\bibitem[EN16]{EN16}
M.~Elkin and O.~Neiman.
\newblock Hopsets with constant hopbound, and applications to approximate
  shortest paths.
\newblock In {\em {IEEE} 57th Annual Symposium on Foundations of Computer
  Science, {FOCS} 2016, 9-11 October 2016, Hyatt Regency, New Brunswick, New
  Jersey, {USA}}, pages 128--137,
\newblock 2016, \href {http://dx.doi.org/10.1109/FOCS.2016.22}
  {\path{doi:10.1109/FOCS.2016.22}}.

\bibitem[EN17a]{EN17}
M.~Elkin and O.~Neiman.
\newblock Efficient algorithms for constructing very sparse spanners and
  emulators.
\newblock In {\em Proceedings of the Twenty-Eighth Annual {ACM-SIAM} Symposium
  on Discrete Algorithms, {SODA} 2017, Barcelona, Spain, Hotel Porta Fira,
  January 16-19}, pages 652--669,
\newblock 2017, \href {http://dx.doi.org/10.1137/1.9781611974782.41}
  {\path{doi:10.1137/1.9781611974782.41}}.

\bibitem[EN17b]{EN17spanner}
M.~Elkin and O.~Neiman.
\newblock Efficient algorithms for constructing very sparse spanners and
  emulators.
\newblock {\em CoRR}, abs/1607.08337, 2017.
\newblock
\newblock Version 2, \href {http://arxiv.org/abs/1607.08337}
  {\path{arXiv:1607.08337}}.

\bibitem[EN18]{EN18}
M.~Elkin and O.~Neiman.
\newblock Near-optimal distributed routing with low memory.
\newblock In {\em Proceedings of the 2018 ACM Symposium on Principles of
  Distributed Computing}, PODC '18, pages 207--216, New York, NY, USA, 2018.
\newblock ACM, \href {http://dx.doi.org/10.1145/3212734.3212761}
  {\path{doi:10.1145/3212734.3212761}}.

\bibitem[ENS15]{ENS15}
M.~Elkin, O.~Neiman, and S.~Solomon.
\newblock Light spanners.
\newblock {\em {SIAM} J. Discrete Math.}, 29(3):1312--1321,
\newblock 2015, \href {http://dx.doi.org/10.1137/140979538}
  {\path{doi:10.1137/140979538}}.

\bibitem[EP01]{EP01}
M.~Elkin and D.~Peleg.
\newblock (1+epsilon, beta)-spanner constructions for general graphs.
\newblock In {\em Proceedings on 33rd Annual {ACM} Symposium on Theory of
  Computing, July 6-8, 2001, Heraklion, Crete, Greece}, pages 173--182,
\newblock 2001, \href {http://dx.doi.org/10.1145/380752.380797}
  {\path{doi:10.1145/380752.380797}}.

\bibitem[ES16]{ES16}
M.~Elkin and S.~Solomon.
\newblock Fast constructions of lightweight spanners for general graphs.
\newblock {\em {ACM} Trans. Algorithms}, 12(3):29:1--29:21, 2016.
\newblock
\newblock See also SODA'13, \href {http://dx.doi.org/10.1145/2836167}
  {\path{doi:10.1145/2836167}}.

\bibitem[EZ06]{EZ06}
M.~Elkin and J.~Zhang.
\newblock Efficient algorithms for constructing (1+epsilon, beta)-spanners in
  the distributed and streaming models.
\newblock {\em Distributed Computing}, 18(5):375--385,
\newblock 2006, \href {http://dx.doi.org/10.1007/s00446-005-0147-2}
  {\path{doi:10.1007/s00446-005-0147-2}}.

\bibitem[FL16]{FL16}
S.~Friedrichs and C.~Lenzen.
\newblock Parallel metric tree embedding based on an algebraic view on
  moore-bellman-ford.
\newblock In {\em Proceedings of the 28th {ACM} Symposium on Parallelism in
  Algorithms and Architectures, {SPAA} 2016, Asilomar State Beach/Pacific
  Grove, CA, USA, July 11-13, 2016}, pages 455--466,
\newblock 2016, \href {http://dx.doi.org/10.1145/2935764.2935777}
  {\path{doi:10.1145/2935764.2935777}}.

\bibitem[FN18]{FN18}
A.~Filtser and O.~Neiman.
\newblock Light spanners for high dimensional norms via stochastic
  decompositions.
\newblock In {\em 26th Annual European Symposium on Algorithms, {ESA} 2018,
  August 20-22, 2018, Helsinki, Finland}, pages 29:1--29:15,
\newblock 2018, \href {http://dx.doi.org/10.4230/LIPIcs.ESA.2018.29}
  {\path{doi:10.4230/LIPIcs.ESA.2018.29}}.

\bibitem[FS16]{FS16}
A.~Filtser and S.~Solomon.
\newblock The greedy spanner is existentially optimal.
\newblock In {\em PODC 2016}, pages 9--17,
\newblock 2016, \href {http://dx.doi.org/10.1145/2933057.2933114}
  {\path{doi:10.1145/2933057.2933114}}.

\bibitem[Gha16]{Gh16}
M.~Ghaffari.
\newblock An improved distributed algorithm for maximal independent set.
\newblock In {\em Proceedings of the Twenty-Seventh Annual {ACM-SIAM} Symposium
  on Discrete Algorithms, {SODA} 2016, Arlington, VA, USA, January 10-12,
  2016}, pages 270--277,
\newblock 2016, \href {http://dx.doi.org/10.1137/1.9781611974331.ch20}
  {\path{doi:10.1137/1.9781611974331.ch20}}.

\bibitem[GKL03]{GKL03}
A.~Gupta, R.~Krauthgamer, and J.~R. Lee.
\newblock Bounded geometries, fractals, and low-distortion embeddings.
\newblock In {\em 44th Symposium on Foundations of Computer Science {(FOCS}
  2003), 11-14 October 2003, Cambridge, MA, USA, Proceedings}, pages 534--543,
\newblock 2003, \href {http://dx.doi.org/10.1109/SFCS.2003.1238226}
  {\path{doi:10.1109/SFCS.2003.1238226}}.

\bibitem[GL14]{GL14}
M.~Ghaffari and C.~Lenzen.
\newblock Near-optimal distributed tree embedding.
\newblock In {\em Distributed Computing - 28th International Symposium, {DISC}
  2014, Austin, TX, USA, October 12-15, 2014. Proceedings}, pages 197--211,
\newblock 2014, \href {http://dx.doi.org/10.1007/978-3-662-45174-8\_14}
  {\path{doi:10.1007/978-3-662-45174-8\_14}}.

\bibitem[GL18]{GL18}
M.~Ghaffari and J.~Li.
\newblock Improved distributed algorithms for exact shortest paths.
\newblock In {\em Proceedings of the 50th Annual {ACM} {SIGACT} Symposium on
  Theory of Computing, {STOC} 2018, Los Angeles, CA, USA, June 25-29, 2018},
  pages 431--444,
\newblock 2018, \href {http://dx.doi.org/10.1145/3188745.3188948}
  {\path{doi:10.1145/3188745.3188948}}.

\bibitem[Got15]{G15}
L.~Gottlieb.
\newblock A light metric spanner.
\newblock In {\em {IEEE} 56th Annual Symposium on Foundations of Computer
  Science, {FOCS} 2015, Berkeley, CA, USA, 17-20 October, 2015}, pages
  759--772,
\newblock 2015, \href {http://dx.doi.org/10.1109/FOCS.2015.52}
  {\path{doi:10.1109/FOCS.2015.52}}.

\bibitem[HM06]{HPM06}
S.~Har{-}Peled and M.~Mendel.
\newblock Fast construction of nets in low-dimensional metrics and their
  applications.
\newblock {\em {SIAM} J. Comput.}, 35(5):1148--1184,
\newblock 2006, \href {http://dx.doi.org/10.1137/S0097539704446281}
  {\path{doi:10.1137/S0097539704446281}}.

\bibitem[KKM{\etalchar{+}}12]{KKMPT12}
M.~Khan, F.~Kuhn, D.~Malkhi, G.~Pandurangan, and K.~Talwar.
\newblock Efficient distributed approximation algorithms via probabilistic tree
  embeddings.
\newblock {\em Distributed Computing}, 25(3):189--205,
\newblock 2012, \href {http://dx.doi.org/10.1007/s00446-012-0157-9}
  {\path{doi:10.1007/s00446-012-0157-9}}.

\bibitem[Kle05]{K05}
P.~N. Klein.
\newblock A linear-time approximation scheme for planar weighted {TSP}.
\newblock In {\em 46th Annual {IEEE} Symposium on Foundations of Computer
  Science {(FOCS} 2005), 23-25 October 2005, Pittsburgh, PA, USA, Proceedings},
  pages 647--657,
\newblock 2005, \href {http://dx.doi.org/10.1109/SFCS.2005.7}
  {\path{doi:10.1109/SFCS.2005.7}}.

\bibitem[KP98]{KP98}
S.~Kutten and D.~Peleg.
\newblock Fast distributed construction of small \emph{k}-dominating sets and
  applications.
\newblock {\em J. Algorithms}, 28(1):40--66,
\newblock 1998, \href {http://dx.doi.org/10.1006/jagm.1998.0929}
  {\path{doi:10.1006/jagm.1998.0929}}.

\bibitem[KRY95]{KRY95}
S.~Khuller, B.~Raghavachari, and N.~E. Young.
\newblock Balancing minimum spanning trees and shortest-path trees.
\newblock {\em Algorithmica}, 14(4):305--321,
\newblock 1995, \href {http://dx.doi.org/10.1007/BF01294129}
  {\path{doi:10.1007/BF01294129}}.

\bibitem[Lub86]{L86}
M.~Luby.
\newblock A simple parallel algorithm for the maximal independent set problem.
\newblock {\em SIAM J. Comput.}, 15(4):1036--1055,
\newblock November 1986, \href {http://dx.doi.org/10.1137/0215074}
  {\path{doi:10.1137/0215074}}.

\bibitem[MP98]{MP98}
Y.~Mansour and D.~Peleg.
\newblock An approximation algorithm for minimum-cost network design.
\newblock Technical report, Weizmann Institute of Science, Rehovot,
\newblock 1998.

\bibitem[MPVX15]{MPVX15}
G.~L. Miller, R.~Peng, A.~Vladu, and S.~C. Xu.
\newblock Improved parallel algorithms for spanners and hopsets.
\newblock In {\em Proc. of 27th SPAA}, pages 192--201,
\newblock 2015, \href {http://dx.doi.org/10.1145/2755573.2755574}
  {\path{doi:10.1145/2755573.2755574}}.

\bibitem[MRSZ11]{MRSZ11}
Y.~M{\'{e}}tivier, J.~M. Robson, N.~Saheb{-}Djahromi, and A.~Zemmari.
\newblock An optimal bit complexity randomized distributed {MIS} algorithm.
\newblock {\em Distributed Computing}, 23(5-6):331--340,
\newblock 2011, \href {http://dx.doi.org/10.1007/s00446-010-0121-5}
  {\path{doi:10.1007/s00446-010-0121-5}}.

\bibitem[NZ02]{NZ02}
T.~S.~E. Ng and H.~Zhang.
\newblock Predicting internet network distance with coordinates-based
  approaches.
\newblock In {\em 21st Annual Joint Conference of the IEEE Computer and
  Communications Societies (INFOCOM)}, pages 178--187,
\newblock 2002, \href {http://dx.doi.org/10.1109/INFCOM.2002.1019258}
  {\path{doi:10.1109/INFCOM.2002.1019258}}.

\bibitem[Pel00]{P00}
D.~Peleg.
\newblock {\em Distributed Computing: A Locality-Sensitive Approach}.
\newblock Society for Industrial and Applied Mathematics,
\newblock 2000, \href {http://dx.doi.org/10.1137/1.9780898719772}
  {\path{doi:10.1137/1.9780898719772}}.

\bibitem[Pet09]{Pet09}
S.~Pettie.
\newblock Low distortion spanners.
\newblock {\em {ACM} Trans. Algorithms}, 6(1):7:1--7:22,
\newblock 2009, \href {http://dx.doi.org/10.1145/1644015.1644022}
  {\path{doi:10.1145/1644015.1644022}}.

\bibitem[Pet10]{P10}
S.~Pettie.
\newblock Distributed algorithms for ultrasparse spanners and linear size
  skeletons.
\newblock {\em Distributed Computing}, 22(3):147--166,
\newblock 2010, \href {http://dx.doi.org/10.1007/s00446-009-0091-7}
  {\path{doi:10.1007/s00446-009-0091-7}}.

\bibitem[PS89]{PS89}
D.~Peleg and A.~A. Sch{\"{a}}ffer.
\newblock Graph spanners.
\newblock {\em Journal of Graph Theory}, 13(1):99--116,
\newblock 1989, \href {http://dx.doi.org/10.1002/jgt.3190130114}
  {\path{doi:10.1002/jgt.3190130114}}.

\bibitem[PU89]{PU89}
D.~Peleg and J.~D. Ullman.
\newblock An optimal synchronizer for the hypercube.
\newblock {\em {SIAM} J. Comput.}, 18(4):740--747,
\newblock 1989, \href {http://dx.doi.org/10.1137/0218050}
  {\path{doi:10.1137/0218050}}.

\bibitem[PV04]{PV04}
P.~Penna and C.~Ventre.
\newblock Energy-efficient broadcasting in ad-hoc networks: combining msts with
  shortest-path trees.
\newblock In {\em Proceedings of the 1st {ACM} International Workshop on
  Performance Evaluation of Wireless Ad Hoc, Sensor, and Ubiquitous Networks,
  {PE-WASUN} 2004, Venezia, Italy, October 4, 2004}, pages 61--68,
\newblock 2004, \href {http://dx.doi.org/10.1145/1023756.1023769}
  {\path{doi:10.1145/1023756.1023769}}.

\bibitem[SCRS01]{SCRS01}
F.~S. Salman, J.~Cheriyan, R.~Ravi, and S.~Subramanian.
\newblock Approximating the single-sink link-installation problem in network
  design.
\newblock {\em {SIAM} Journal on Optimization}, 11(3):595--610,
\newblock 2001, \href {http://dx.doi.org/10.1137/S1052623497321432}
  {\path{doi:10.1137/S1052623497321432}}.

\bibitem[SEW13]{SEW13}
J.~Schneider, M.~Elkin, and R.~Wattenhofer.
\newblock Symmetry breaking depending on the chromatic number or the
  neighborhood growth.
\newblock {\em Theor. Comput. Sci.}, 509:40--50,
\newblock 2013, \href {http://dx.doi.org/10.1016/j.tcs.2012.09.004}
  {\path{doi:10.1016/j.tcs.2012.09.004}}.

\bibitem[SHK{\etalchar{+}}12]{SHKKNPPW12}
A.~D. Sarma, S.~Holzer, L.~Kor, A.~Korman, D.~Nanongkai, G.~Pandurangan,
  D.~Peleg, and R.~Wattenhofer.
\newblock Distributed verification and hardness of distributed approximation.
\newblock {\em {SIAM} J. Comput.}, 41(5):1235--1265,
\newblock 2012, \href {http://dx.doi.org/10.1137/11085178X}
  {\path{doi:10.1137/11085178X}}.

\bibitem[TSL00]{TSL00}
J.~B. Tenenbaum, V.~d. Silva, and J.~C. Langford.
\newblock A global geometric framework for nonlinear dimensionality reduction.
\newblock {\em Science}, 290(5500):2319--2323,
\newblock 2000, \href {http://dx.doi.org/10.1126/science.290.5500.2319}
  {\path{doi:10.1126/science.290.5500.2319}}.

\bibitem[TZ06]{TZ06}
M.~Thorup and U.~Zwick.
\newblock Spanners and emulators with sublinear distance errors.
\newblock In {\em Proceedings of the Seventeenth Annual {ACM-SIAM} Symposium on
  Discrete Algorithms, {SODA} 2006, Miami, Florida, USA, January 22-26, 2006},
  pages 802--809, 2006.
\newblock
\newblock \url{http://dl.acm.org/citation.cfm?id=1109557.1109645}.

\bibitem[WCT02]{WCT02}
B.~Y. Wu, K.-M. Chao, and C.~Y. Tang.
\newblock Light graphs with small routing cost.
\newblock {\em Networks}, 39(3):130--138,
\newblock 2002, \href {http://dx.doi.org/10.1002/net.10019}
  {\path{doi:10.1002/net.10019}}.

\bibitem[YCC06]{YCC06}
J.~Yu, L.~Chen, and G.~Chen.
\newblock Priority based overlay multicast with filtering mechanism for
  distributed interactive applications.
\newblock In {\em 10th {IEEE} International Symposium on Distributed Simulation
  and Real-Time Applications {(DS-RT} 2006), 2-4 October 2006, Malaga, Spain},
  pages 127--134,
\newblock 2006, \href {http://dx.doi.org/10.1109/DS-RT.2006.29}
  {\path{doi:10.1109/DS-RT.2006.29}}.

\end{thebibliography}

\end{document}